\documentclass[11pt]{article}%
\usepackage{amsmath}
\usepackage{amsfonts}
\usepackage{amssymb}
\usepackage{graphicx}%
\newtheorem{theorem}{Theorem}

\newtheorem{corollary}[theorem]{Corollary}

\newtheorem{definition}[theorem]{Definition}
\newtheorem{example}[theorem]{Example}

\newtheorem{lemma}[theorem]{Lemma}

\newtheorem{proposition}[theorem]{Proposition}
\newtheorem{remark}[theorem]{Remark}

\newenvironment{proof}[1][Proof]{\noindent\textbf{#1.} }{\ \rule{0.5em}{0.5em}}
\begin{document}

\title{Vector-valued Jack Polynomials and Wavefunctions on the Torus}
\author{Charles F. Dunkl\\Department of Mathematics, University of Virginia,\\Charlottesville, VA 22904-4137, U.S.}
\date{7 February 2017}
\maketitle

\begin{abstract}
The Hamiltonian of the quantum Calogero-Sutherland model of $N$ identical
particles on the circle with $1/r^{2}$ interactions has eigenfunctions
consisting of Jack polynomials times the base state. By use of the generalized
Jack polynomials taking values in modules of the symmetric group and the
matrix solution of a system of linear differential equations one constructs
novel eigenfunctions of the Hamiltonian. Like the usual wavefunctions each
eigenfunction determines a symmetric probability density on the $N$-torus. The
construction applies to any irreducible representation of the symmetric group.
The methods depend on the theory of generalized Jack polynomials due to
Griffeth, and the Yang-Baxter graph approach of Luque and the author.

\end{abstract}

\section{Introduction}

The quantum Calogero-Sutherland model for $N$ identical particles with
$1/r^{2}$ interactions on the unit circle has the Hamiltonian%
\begin{align*}
\mathcal{H} &  =\mathcal{-}\sum_{i=1}^{N}\left(  \frac{\partial}%
{\partial\theta_{i}}\right)  ^{2}+\frac{1}{2}\sum_{1\leq i<j\leq N}%
\frac{\kappa\left(  \kappa-1\right)  }{\sin^{2}\left(  \frac{1}{2}\left(
\theta_{i}-\theta_{j}\right)  \right)  }\\
&  =\sum_{i=1}^{N}\left(  x_{i}\frac{\partial}{\partial x_{i}}\right)
^{2}-2\kappa\sum_{1\leq i<j\leq N}\frac{x_{i}x_{j}\left(  \kappa-1\right)
}{\left(  x_{i}-x_{j}\right)  ^{2}},
\end{align*}
where $x_{j}=e^{\mathrm{i}\theta_{j}}$ and $-\pi<\theta_{j}\leq\pi$ for $1\leq
j\leq N$. The time-independent Schr\"{o}dinger equation $\mathcal{H}\psi
=E\psi$ has solutions expressible as the product of the base-state%
\[
\psi_{0}\left(  x\right)  =\prod\limits_{1\leq i<j\leq N}\left(
-\frac{\left(  x_{i}-x_{j}\right)  ^{2}}{x_{i}x_{j}}\right)  ^{\kappa/2}%
\]
with a Jack polynomial (Lapointe and Vinet \cite{LV1996}, Awata \cite{A1997}).
The base-state is a solution of the first-order linear differential system%
\[
\frac{\partial}{\partial x_{i}}\psi_{0}\left(  x\right)  =\kappa\psi
_{0}\left(  x\right)  \left\{  \sum_{j\neq i}\frac{1}{x_{i}-x_{j}}-\frac
{N-1}{2x_{i}}\right\}  ,1\leq i\leq N;
\]
and $\mathcal{H}\psi_{0}=\frac{1}{12}\kappa^{2}N\left(  N^{2}-1\right)
\psi_{0}$. The theory of Jack polynomials has been generalized to polynomials
taking values in modules of the symmetric group (Griffeth \cite{G2010}). In
this paper the Hamiltonian $\mathcal{H}$ will be interpreted in that context.
The base state $\psi_{0}$ is replaced by a matrix function satisfying an
analogous differential system and the generalized wavefunctions are
vector-valued. Nevertheless for an interval of parameter values depending on
the module the wavefunctions do give rise to symmetric probability density
functions on the torus. The interval is symmetric about $\kappa=0$ hence this
is qualitatively different from the usual scalar case where $\kappa$ is
unbounded above.

Section \ref{JPops} is a brief overview of representation theory for the
symmetric groups, and the commutative set of operators on polynomials of which
the nonsymmetric Jack polynomials are simultaneous eigenfunctions. Section
\ref{mtxbas} concerns the first-order linear differential system defining the
basic matrix function needed to map the polynomials to eigenfunctions of the
Hamiltonian modified with twisted exchange operators. In Section \ref{Hforms}
there is a description of the Hermitian form related to integration of
vector-valued polynomials on the torus, and the Yang-Baxter graph technique
for constructing the nonsymmetric Jack polynomials. Section \ref{symmvp}
presents the adaptation of the method of Baker and Forrester \cite{B1999} to
form symmetric Jack polynomials from the nonsymmetric polynomials; the
analysis involves tableaux with certain properties. Also this section contains
the formulae for the squared norms of the Jack polynomials. Then Section
\ref{symmwav} uses the vector-valued Jack polynomials and the matrix function
from Section \ref{mtxbas} to construct vector-valued eigenfunctions of the
Hamiltonian $\mathcal{H}$ and the associated probability density. Also the
Jack polynomial of minimal degree is described, and finally there is a brief
description of the matrix function in the case of the two-dimensional
representation of $\mathcal{S}_{4}$.

\section{\label{JPops}The generalized Jack polynomials and associated
operators}

The \textit{symmetric group} $\mathcal{S}_{N}$, the set of permutations of
$\left\{  1,2,\ldots,N\right\}  $, acts on $\mathbb{C}^{N}$ by permutation of
coordinates. For $\alpha\in\mathbb{Z}^{N}$ the norm is $\left\vert
\alpha\right\vert :=\sum_{i=1}^{N}\left\vert \alpha_{i}\right\vert $ and the
monomial is $x^{\alpha}:=\prod_{i=1}^{N}x_{i}^{\alpha_{i}}$. Denote
$\mathbb{N}_{0}:=\left\{  0,1,2,\ldots\right\}  $. The space of polynomials
$\mathcal{P}:=\mathrm{span}_{\mathbb{C}}\left\{  x^{\alpha}:\alpha
\in\mathbb{N}_{0}^{N}\right\}  $. Elements of $\mathrm{span}_{\mathbb{C}%
}\left\{  x^{\alpha}:\alpha\in\mathbb{Z}^{N}\right\}  $ are called
\textit{Laurent} polynomials. The action of $\mathcal{S}_{N}$ is extended to
polynomials by $wp\left(  x\right)  =p\left(  xw\right)  $ where $\left(
xw\right)  _{i}=x_{w\left(  i\right)  }$ (consider $x$ as a row vector and $w$
as a permutation matrix, $\left[  w\right]  _{ij}=\delta_{i,w\left(  j\right)
}$, then $xw=x\left[  w\right]  $). This is a representation of $\mathcal{S}%
_{N}$, that is, $w_{1}\left(  w_{2}p\right)  \left(  x\right)  =\left(
w_{2}p\right)  \left(  xw_{1}\right)  =p\left(  xw_{1}w_{2}\right)  =\left(
w_{1}w_{2}\right)  p\left(  x\right)  $ for all $w_{1},w_{2}\in\mathcal{S}%
_{N}$.

Furthermore $\mathcal{S}_{N}$ is generated by reflections in the mirrors
$\left\{  x:x_{i}=x_{j}\right\}  $ for $1\leq i<j\leq N$. These are
\textit{transpositions, }denoted by $\left(  i,j\right)  $, interchanging
$x_{i}$ and $x_{j}$. Define the $\mathcal{S}_{N}$-action on $\alpha
\in\mathbb{Z}^{N}$ so that $\left(  xw\right)  ^{\alpha}=x^{w\alpha}$%
\[
\left(  xw\right)  ^{\alpha}=\prod_{i=1}^{N}x_{w\left(  i\right)  }%
^{\alpha_{i}}=\prod_{j=1}^{N}x_{j}^{\alpha_{w^{-1}\left(  j\right)  }},
\]
that is $\left(  w\alpha\right)  _{i}=\alpha_{w^{-1}\left(  i\right)  }$
(consider $\alpha$ as a column vector, then $w\alpha=\left[  w\right]  \alpha$).

The \textit{simple reflections} $s_{i}:=\left(  i,i+1\right)  $, $1\leq i<N$,
suffice to generate $\mathcal{S}_{N}$. They are the key devices for applying
inductive methods, and satisfy the \textit{braid} relations:
\begin{align*}
s_{i}s_{j}  &  =s_{j}s_{i},\left\vert i-j\right\vert \geq2;\\
s_{i}s_{i+1}s_{i}  &  =s_{i+1}s_{i}s_{i+1}.
\end{align*}

We consider the situation where the group $\mathcal{S}_{N}$ acts on the range
as well as on the domain of the polynomials. We use vector spaces, called
$\mathcal{S}_{N}$-modules, on which $\mathcal{S}_{N}$ has an irreducible
unitary (orthogonal) representation$:\tau:\mathcal{S}_{N}\rightarrow
O_{m}\left(  \mathbb{R}\right)  $ ($\tau\left(  w\right)  ^{-1}=\tau\left(
w^{-1}\right)  =\tau\left(  w\right)  ^{T}$). See James and Kerber
\cite{J1981} for representation theory and a modern discussion of Young's methods.

Denote the set of partitions $\mathbb{N}_{0}^{N,+}=\left\{  \lambda
\in\mathbb{N}_{0}^{N}:\lambda_{1}\geq\lambda_{2}\geq\cdots\geq\lambda
_{N}\right\}  $. We identify $\tau$ with a partition of $N$ given the same
label, that is $\tau\in\mathbb{N}_{0}^{N,+}$ and $\left\vert \tau\right\vert
=N$. The length of $\tau$ is $\ell\left(  \tau\right)  =\max\left\{
i:\tau_{i}>0\right\}  $. There is a Ferrers diagram of shape $\tau$ (also
given the same label), with boxes at points $\left(  i,j\right)  $ with $1\leq
i\leq\ell\left(  \tau\right)  $ and $1\leq j\leq\tau_{i}$. A \textit{tableau}
of shape $\tau$ is a filling of the boxes with numbers, and a \textit{reverse
standard Young tableau} (RSYT) is a filling with the numbers $\left\{
1,2,\ldots,N\right\}  $ so that the entries decrease in each row and each
column. We require $\dim V_{\tau}\geq2$, thus excluding the one-dimensional
representations corresponding to one-row $\left(  N\right)  $ or one-column
$\left(  1,1,\ldots,1\right)  $ partitions (the trivial and determinant
representations, respectively). The \textit{hook-length} of the node $\left(
i,j\right)  \in\tau$ is%
\begin{align}
h\left(  i,j\right)   &  :=\tau_{i}-j+\#\left\{  k:j\leq\tau_{k},i<k\leq
\ell\left(  \tau\right)  \right\}  +1,\label{hllt}\\
h_{\tau}  &  :=h\left(  1,1\right)  =\tau_{1}+\ell\left(  \tau\right)
-1\nonumber
\end{align}
and $h_{\tau}$ is the maximum hook-length of $\tau$. Denote the set of RSYT's
of shape $\tau$ by $\mathcal{Y}\left(  \tau\right)  $ and let $V_{\tau
}=\mathrm{span}_{\mathbb{C}}\left\{  T:T\in\mathcal{Y}\left(  \tau\right)
\right\}  $ with orthogonal basis $\mathcal{Y}\left(  \tau\right)  $. The
formulae for the action of $s_{i}$ on $\mathcal{Y}\left(  \tau\right)  $ are
described in Proposition \ref{sympol} below. The \textit{hook-length formula}
is $\#\mathcal{Y}\left(  \tau\right)  =N!/\prod\limits_{\left(  i,j\right)
\in\tau}h\left(  i,j\right)  $. Set $n_{\tau}:=\dim V_{\tau}=\#\mathcal{Y}%
\left(  \tau\right)  $. For $1\leq i\leq N$ and $T\in\mathcal{Y}\left(
\tau\right)  $ the entry $i$ is at coordinates $\left(  \mathrm{rw}\left(
i,T\right)  ,\mathrm{cm}\left(  i,T\right)  \right)  $ and the
\textit{content} is $c\left(  i,T\right)  :=\mathrm{cm}\left(  i,T\right)
-\mathrm{rw}\left(  i,T\right)  $. Each $T\in\mathcal{Y}\left(  \tau\right)  $
is uniquely determined by its \textit{content vector} $\left[  c\left(
i,T\right)  \right]  _{i=1}^{N}$. Let $S_{1}\left(  \tau\right)  :=\sum
_{i=1}^{N}c\left(  i,T\right)  $ (this sum depends only on $\tau$) and
$\gamma:=S_{1}\left(  \tau\right)  /N$. The $\mathcal{S}_{N}$-invariant inner
product on $V_{\tau}$ is defined by%
\begin{equation}
\left\langle T,T^{\prime}\right\rangle _{0}:=\delta_{T,T^{\prime}}\times
\prod_{\substack{1\leq i<j\leq N,\\c\left(  i,T\right)  \leq c\left(
j,T\right)  -2}}\left(  1-\frac{1}{\left(  c\left(  i,T\right)  -c\left(
j,T\right)  \right)  ^{2}}\right)  ,~T,T^{\prime}\in\mathcal{Y}\left(
\tau\right)  . \label{Tnorm}%
\end{equation}
It is unique up to multiplication by a constant.

The \textit{Jucys-Murphy} elements $\omega_{i}:=\sum\limits_{j=i+1}^{N}\left(
i,j\right)  $ satisfy $\sum\limits_{j=i+1}^{N}\tau\left(  i,j\right)
T=c\left(  i,T\right)  T$ and thus the central element $\sum\limits_{1\leq
i<j\leq N}\left(  i,j\right)  $ (in the group algebra $\mathbb{R}%
\mathcal{S}_{N}$) satisfies $\sum\limits_{1\leq i<j\leq N}\tau\left(
i,j\right)  T=S_{1}\left(  \tau\right)  T$ for each $T\in\mathcal{Y}\left(
\tau\right)  $. We abbreviate $\tau\left(  \left(  i,j\right)  \right)  $ to
$\tau\left(  i,j\right)  .$

The generalized Jack polynomials are elements of $\mathcal{P}_{\tau
}=\mathcal{P}\otimes V_{\tau}$, the space of $V_{\tau}$-valued polynomials,
which is equipped with the $\mathcal{S}_{N}$ action:%
\begin{align*}
w\left(  x^{\alpha}\otimes T\right)   &  =\left(  xw\right)  ^{\alpha}%
\otimes\tau\left(  w\right)  T,~\alpha\in\mathbb{N}_{0}^{N},T\in
\mathcal{Y}\left(  \tau\right)  ,\\
wp\left(  x\right)   &  =\tau\left(  w\right)  p\left(  xw\right)
,~p\in\mathcal{P}_{\tau},
\end{align*}
extended by linearity. A symmetric polynomial $p$ satisfies $wp=p$, that is,
$p\left(  xw\right)  =\tau\left(  w\right)  ^{-1}p\left(  x\right)  $ for all
$w\in\mathcal{S}_{N}$. . The following describes the transformation rules for
$\tau\left(  s_{i}\right)  $ acting on $\mathcal{Y}\left(  \tau\right)  $ and
on symmetric polynomials.

\begin{proposition}
\label{sympol}Suppose $p\in\mathcal{P}_{\tau}$ is symmetric, $1\leq i<N$, and
$T\in\mathcal{Y}\left(  \tau\right)  $. Express $p\left(  x\right)
=\sum\limits_{T^{\prime}\in\mathcal{Y}\left(  \tau\right)  }\frac
{1}{\left\langle T^{\prime},T^{\prime}\right\rangle _{0}}p_{T^{\prime}}\left(
x\right)  \otimes T^{\prime}.$ If $c\left(  i,T\right)  =c\left(
i+1,T\right)  +1$ then $\tau\left(  s_{i}\right)  T=T$ and $s_{i}p_{T}=p_{T}$;
if $c\left(  i,T\right)  =c\left(  i+1,T\right)  -1$ then $\tau\left(
s_{i}\right)  T=-T$ and $s_{i}p_{T}=-p_{T}$; if $c\left(  i,T\right)
-c\left(  i+1,T\right)  \geq2$ and $T^{\left(  i\right)  }$ is $T$ with
$i,i+1$ interchanged then $\tau\left(  s_{i}\right)  T=T^{\left(  i\right)
}+bT$ and $s_{i}p_{T}=p_{T^{\left(  i\right)  }}+bp_{T}$, where $b=\frac
{1}{c\left(  i,T\right)  -c\left(  i+1,T\right)  }$.
\end{proposition}

\begin{proof}
The transformation properties of $p_{T}$ follow from $s_{i}\left(
p_{T}\left(  x\right)  \otimes T\right)  =p\left(  xs_{i}\right)  \otimes
\tau\left(  s_{i}\right)  T$. The first case is when $\mathrm{rw}\left(
i,T\right)  =\mathrm{rw}\left(  i+1,T\right)  $ and $\tau\left(  s_{i}\right)
T=T$; the second case is when $\mathrm{cm}\left(  i,T\right)  =\mathrm{cm}%
\left(  i+1,T\right)  $ and $\tau\left(  s_{i}\right)  T=-T$. In the case
$c\left(  i,T\right)  -c\left(  i+1,T\right)  \geq2$ the relations
$\tau\left(  s_{i}\right)  T^{\left(  i\right)  }=\left(  1-b^{2}\right)
T-bT^{\left(  i\right)  }$ and $\left\langle T^{\left(  i\right)  },T^{\left(
i\right)  }\right\rangle _{0}=\left(  1-b^{2}\right)  \left\langle
T,T\right\rangle _{0}$ hold. By hypothesis%
\begin{align*}
&  \frac{1}{\left\langle T,T\right\rangle _{0}}p_{T}\left(  x\right)  \otimes
T+\frac{1}{\left\langle T^{\left(  i\right)  },T^{\left(  i\right)
}\right\rangle _{0}}p_{T^{\left(  i\right)  }}\left(  x\right)  \otimes
T^{\left(  i\right)  }\\
&  =\frac{1}{\left\langle T,T\right\rangle _{0}}p_{T}\left(  xs_{i}\right)
\otimes\tau\left(  s_{i}\right)  T+\frac{1}{\left\langle T^{\left(  i\right)
},T^{\left(  i\right)  }\right\rangle _{0}}p_{T^{\left(  i\right)  }}\left(
xs_{i}\right)  \otimes\tau\left(  s_{i}\right)  T^{\left(  i\right)  }\\
&  =\frac{1}{\left\langle T,T\right\rangle _{0}}\left\{  p_{T}\left(
xs_{i}\right)  \otimes\left(  T^{\left(  i\right)  }+bT\right)  +\frac
{1}{1-b^{2}}p_{T^{\left(  i\right)  }}\left(  xs_{i}\right)  \otimes\left(
\left(  1-b^{2}\right)  T-bT^{\left(  i\right)  }\right)  \right\} \\
&  =\frac{1}{\left\langle T,T\right\rangle _{0}}\left\{  \left(  bp_{T}\left(
xs_{i}\right)  +p_{T^{\left(  i\right)  }}\left(  xs_{i}\right)  \right)
\otimes T+\left(  \left(  1-b^{2}\right)  p_{T}\left(  xs_{i}\right)
-bp_{T^{\left(  i\right)  }}\left(  xs_{i}\right)  \right)  \otimes T^{\left(
i\right)  }\right\}  .
\end{align*}
Replace $x$ by $xs_{i}$ and conclude that $s_{i}p_{T}=p_{T^{\left(  i\right)
}}+bp_{T}$ and $s_{i}p_{T^{\left(  i\right)  }}=\left(  1-b^{2}\right)
p_{T}-bp_{T^{\left(  i\right)  }}$.
\end{proof}

The polynomials $p_{T}$ can be derived from $p_{T_{0}}$ where $T_{0}$ is the
root RSYT (with $N,N-1,\ldots$ entered column by column), but determining
which polynomials can serve as $p_{T_{0}}$ is nontrivial in general.

There is a parameter $\kappa\in\mathbb{R}$ (in general, $\kappa$ could be transcendental).

\begin{definition}
The \textit{Dunkl} and \textit{Cherednik-Dunkl} operators are ($1\leq i\leq
N,p\in\mathcal{P}_{\tau}$)
\begin{align*}
\mathcal{D}_{i}p\left(  x\right)   &  :=\frac{\partial}{\partial x_{i}%
}p\left(  x\right)  +\kappa\sum_{j\neq i}\tau\left(  i,j\right)
\frac{p\left(  x\right)  -p\left(  x\left(  i,j\right)  \right)  }{x_{i}%
-x_{j}},\\
\mathcal{U}_{i}p\left(  x\right)   &  :=\mathcal{D}_{i}\left(  x_{i}p\left(
x\right)  \right)  -\kappa\sum_{j=1}^{i-1}\tau\left(  i,j\right)  p\left(
x\left(  i,j\right)  \right)  .
\end{align*}

\end{definition}

The commutation relations analogous to the scalar case hold:%

\begin{align}
\mathcal{D}_{i}\mathcal{D}_{j}  &  =\mathcal{D}_{j}\mathcal{D}_{i}%
,~\mathcal{U}_{i}\mathcal{U}_{j}=\mathcal{U}_{j}\mathcal{U}_{i},~1\leq i,j\leq
N\label{commU}\\
w\mathcal{D}_{i}  &  =\mathcal{D}_{w\left(  i\right)  }w,\forall
w\in\mathcal{S}_{N};~s_{j}\mathcal{U}_{i}=\mathcal{U}_{i}s_{j},~j\neq
i-1,i;\nonumber\\
s_{i}\mathcal{U}_{i}s_{i}  &  =\mathcal{U}_{i+1}+\kappa s_{i},~\mathcal{U}%
_{i}s_{i}=s_{i}\mathcal{U}_{i+1}+\kappa,~\mathcal{U}_{i+1}s_{i}=s_{i}%
\mathcal{U}_{i}-\kappa.\nonumber
\end{align}

The commutation properties for the $\mathcal{U}_{i}$ and $s_{j}$ are derived
as follows:

\begin{enumerate}
\item if $j>i$ then $s_{j}$ commutes with each term in $\mathcal{U}_{i}$,

\item if $j<i-1$ then $\mathcal{D}_{i}x_{i}s_{j}-\kappa\sum_{k<i}\left(
k,i\right)  s_{j}=s_{j}\mathcal{D}_{i}x_{i}-\kappa\sum_{k<i}s_{j}\left(
k,i\right)  $ because $\left(  k,i\right)  s_{j}=s_{j}\left(  k,i\right)  $
unless $k=j$ or $j+1$ and then $\left\{  \left(  i,j\right)  +\left(
i,j+1\right)  \right\}  s_{j}=s_{j}\left\{  \left(  i,j+1\right)  +\left(
i,j\right)  \right\}  $,

\item if $j=i$ then $s_{i}\mathcal{U}_{i}s_{i}=\mathcal{D}_{i+1}x_{i+1}%
-\kappa\sum_{k<i}s_{i}\left(  k,i\right)  s_{i}=\mathcal{D}_{i+1}%
x_{i+1}-\kappa\sum_{k<i}\left(  k,i+1\right)  =\mathcal{U}_{i+1}+\kappa\left(
i,i+1\right)  $; the relations $\mathcal{U}_{i}s_{i}=s_{i}\mathcal{U}%
_{i+1}+\kappa,~\mathcal{U}_{i+1}s_{i}=s_{i}\mathcal{U}_{i}-\kappa$ follow from
right and left multiplication by $s_{i}$.
\end{enumerate}

From the commutation $\mathcal{D}_{i}x_{i}-x_{i}\mathcal{D}_{i}=1+\kappa
\sum_{j\neq i}\left(  i,j\right)  $ we obtain%
\begin{equation}
\mathcal{U}_{i}=\mathcal{D}_{i}x_{i}-\kappa\sum_{j<i}\left(  i,j\right)
=x_{i}\mathcal{D}_{i}+1+\kappa\sum_{j>i}\left(  i,j\right)  =x_{i}%
\mathcal{D}_{i}+1+\kappa\omega_{i}. \label{UDx}%
\end{equation}

\begin{proposition}
\label{Usym}If $q\left(  x_{1},x_{2},\ldots,x_{N}\right)  $ is a symmetric
polynomial then $q\left(  \mathcal{U}_{1},\mathcal{U}_{2},\ldots
,\mathcal{U}_{N}\right)  $ commutes with each $w\in\mathcal{S}_{N}$, as an
operator on $\mathcal{P}_{\tau}$.
\end{proposition}

\begin{proof}
It suffices to prove the commutativity for each $s_{i}$ with $1\leq i<N$ and
each elementary symmetric polynomial in $\left\{  \mathcal{U}_{i}\right\}  $,
that is, for $\prod\limits_{j=1}^{N}\left(  1+t\mathcal{U}_{j}\right)  $. By
the above formulae it suffices to show $s_{i}$ commutes with $\left(
1+t\mathcal{U}_{i}\right)  \left(  1+t\mathcal{U}_{i+1}\right)  =1+t\left(
\mathcal{U}_{i}+\mathcal{U}_{i+1}\right)  +t^{2}\mathcal{U}_{i}\mathcal{U}%
_{i+1}$. Indeed%
\begin{align*}
s_{i}\left(  \mathcal{U}_{i}+\mathcal{U}_{i+1}\right)  s_{i}  &
=\mathcal{U}_{i+1}+\kappa s_{i}+\mathcal{U}_{i}-\kappa s_{i}=\mathcal{U}%
_{i+1}+\mathcal{U}_{i},\\
s_{i}\mathcal{U}_{i}\mathcal{U}_{i+1}s_{i}  &  =\left(  \mathcal{U}_{i+1}%
s_{i}+\kappa\right)  \left(  s_{i}\mathcal{U}_{i}-\kappa\right)
=\mathcal{U}_{i+1}\mathcal{U}_{i}+\kappa\left(  s_{ii}\mathcal{U}%
_{i}-\mathcal{U}_{i+1}s_{i}-\kappa\right)  =\mathcal{U}_{i+1}\mathcal{U}_{i}.
\end{align*}

\end{proof}

The \textit{nonsymmetric} (vector-valued) \textit{Jack polynomials} (NSJP) are
defined to be simultaneous eigenfunctions of the commuting set $\left\{
\mathcal{U}_{i}:1\leq i\leq N\right\}  $. The symmetric vector-valued Jack
polynomials are simultaneous eigenfunctions of the symmetric polynomials in
$\left\{  \mathcal{U}_{i}\right\}  $. If $p$ is a NSJP then the sum
$\sum\limits_{w\in\mathcal{S}_{N}}wp$ is either a scalar multiple of a
symmetric Jack polynomial or zero. The details are presented in Section
\ref{symmvp}.

\section{\label{mtxbas}The matrix analogue of the base state}

This is a summary of the pertinent results from \cite{D2017}. Vectors and
matrices throughout are of size $n_{\tau}$ and $n_{\tau}\times n_{\tau}$ and
are expressed with respect to the orthonormal basis $\left\{  \left\langle
T,T\right\rangle _{0}^{-1/2}T:T\in\mathcal{Y}\left(  \tau\right)  \right\}  $.
With $\partial_{i}:=\frac{\partial}{\partial x_{i}}$ for $1\leq i\leq N$ the
differential system for the matrix function $L$ is%
\begin{align}
\partial_{i}L\left(  x\right)   &  =\kappa L\left(  x\right)  \left\{
\sum_{j\neq i}\frac{1}{x_{i}-x_{j}}\tau\left(  i,j\right)  -\frac{\gamma
}{x_{i}}I\right\}  ,~1\leq i\leq N,\label{Lsys}\\
\gamma &  :=\frac{S_{1}\left(  \tau\right)  }{N}=\frac{1}{2N}\sum_{i=1}%
^{\ell\left(  \tau\right)  }\tau_{i}\left(  \tau_{i}-2i+1\right)  .\nonumber
\end{align}
The effect of the term $\frac{\gamma}{x_{i}}I$ is to make $L\left(  x\right)
$ homogeneous of degree zero, that is, $\sum_{i=1}^{N}x_{i}\partial
_{i}L\left(  x\right)  =0$. The differential system is defined on
$\mathbb{C}_{reg}^{N}:=\mathbb{C}_{\times}^{N}\backslash\bigcup\limits_{1\leq
i<j\leq N}\left\{  x:x_{i}=x_{j}\right\}  $ (where $\mathbb{C}_{\times
}:=\mathbb{C}\backslash\left\{  0\right\}  $), and it is Frobenius integrable
and analytic, thus any local solution can be continued analytically to any
point in $\mathbb{C}_{reg}^{N}$. The equation is a modified version of the
Knizhnik-Zamolodchikov equation. For the trivial representation $\tau=\left(
N\right)  $ the solution $L\left(  x\right)  =\psi_{0}\left(  x\right)  $ (up
to scalar multiplication) because $\tau\left(  i,j\right)  =I$ and
$\gamma=\frac{N-1}{2}$. The notations for the torus and its surface measure in
terms of polar coordinates are%
\begin{align*}
\mathbb{T}^{N}  &  :=\left\{  x\in\mathbb{C}^{N}:\left\vert x_{j}\right\vert
=1,1\leq j\leq N\right\}  ,\\
\mathrm{d}m\left(  x\right)   &  =\left(  2\pi\right)  ^{-N}\mathrm{d}%
\theta_{1}\cdots\mathrm{d}\theta_{N},~x_{j}=\exp\left(  \mathrm{i}\theta
_{j}\right)  ,-\pi<\theta_{j}\leq\pi,1\leq j\leq N.
\end{align*}

Let $\mathbb{T}_{reg}^{N}:=\mathbb{T}^{N}\cap\mathbb{C}_{reg}^{N}$, then
$\mathbb{T}_{reg}^{N}$ has $\left(  N-1\right)  !$ connected components and
each component is homotopic to a circle; if $x$ is in some component then so
is $ux=\left(  ux_{1},\ldots,ux_{N}\right)  $ for each $u\in\mathbb{T}$.

\begin{definition}
Let $x_{0}:=\left(  1,e^{2\pi\mathrm{i}/N},e^{4\pi\mathrm{i}/N},\ldots
,e^{2\left(  N-1\right)  \pi\mathrm{i}/N}\right)  $ and denote the connected
component of $\mathbb{T}_{reg}^{N}$ containing $x_{0}$ by $\mathcal{C}_{0}$,
called the fundamental chamber.
\end{definition}

Thus $\mathcal{C}_{0}$ is the set consisting of $\left(  e^{\mathrm{i}%
\theta_{1}},\ldots,e^{\mathrm{i}\theta_{N}}\right)  $ with $\theta_{1}%
<\theta_{2}<\cdots<\theta_{N}<\theta_{1}+2\pi$. The homogeneity $L\left(
ux\right)  =L\left(  x\right)  $ for $\left\vert u\right\vert =1$ shows that
$L\left(  x\right)  $ has a well-defined analytic continuation to all of
$\mathcal{C}_{0}$ starting from $x_{0}$. Let $w_{0}:=\left(  1,2,3,\ldots
N\right)  =\left(  12\right)  \left(  23\right)  \cdots\left(  N-1,N\right)
$, an $N$-cycle, and let $\left\langle w_{0}\right\rangle $ denote the cyclic
group generated by $w_{0}$. There are two components of $\mathbb{T}_{reg}^{N}$
which are set-wise invariant under $\left\langle w_{0}\right\rangle $ namely
$\mathcal{C}_{0}$ and the reverse $\left\{  \theta_{N}<\theta_{N-1}%
<\ldots<\theta_{1}<\theta_{N}+2\pi\right\}  $. Indeed $\left\langle
w_{0}\right\rangle $ is the stabilizer of $\mathcal{C}_{0}$ as a subgroup of
$\mathcal{S}_{N}$. A list of properties of $L\left(  x\right)  $ (from
\cite{D2017}):

\begin{enumerate}
\item If $L\left(  x\right)  $ is a solution of (\ref{Lsys}) in some connected
open subset $U$ of $\mathbb{C}_{reg}^{N}$ then $L\left(  xw\right)
\tau\left(  w\right)  ^{-1}$ is a solution in $Uw^{-1}$;

\item If $L\left(  x_{0}\right)  $ is nonsingular then $L$ is nonsingular on
all of $\mathcal{C}_{0}$; this follows from%
\[
\det L\left(  x\right)  =c\prod\limits_{1\leq i<j\leq N}\left(  -\frac{\left(
x_{i}-x_{j}\right)  ^{2}}{x_{i}x_{j}}\right)  ^{\kappa\lambda/2}%
,~\lambda:=\frac{\gamma n_{\tau}}{2\left(  N-1\right)  }=\mathrm{tr}\left(
\tau\left(  1,2\right)  \right)  ,
\]

\item Suppose $L\left(  x\right)  $ is normalized by $L\left(  x_{0}\right)
=I$ then $L\left(  xw_{0}^{m}\right)  =\tau\left(  w_{0}\right)  ^{-m}L\left(
x\right)  \tau\left(  w_{0}\right)  ^{m}$ for all $x\in\mathcal{C}_{0}$ and
$m\in\mathbb{Z}$;
\end{enumerate}

For $w\in\mathcal{S}_{N}$ and for $x\in\mathbb{T}_{reg}^{N}$ define $w_{x}%
\in\mathcal{S}_{N}$ such that $xw_{x}^{-1}\in\mathcal{C}_{0}$ and
$w_{x}\left(  1\right)  =1$; then $w_{x}$ is uniquely defined and is constant
on connected components. Then define $L\left(  x\right)  $ on the other
connected components of $\mathbb{T}_{reg}^{N}$ by%
\begin{equation}
L\left(  x\right)  :=L\left(  xw_{x}^{-1}\right)  \tau\left(  w_{x}\right)  .
\label{TregL}%
\end{equation}
In order to derive a formula for the relation of $L\left(  xw\right)
\tau\left(  w\right)  ^{-1}$ to $L\left(  x\right)  $ we need a twist: for
$w\in\mathcal{S}_{N}$ and for $x\in\mathbb{T}_{reg}^{N}$ define%
\[
M\left(  w,x\right)  :=\tau\left(  w_{0}\right)  ^{1-w_{x}w\left(  1\right)
}.
\]
Henceforth the assumption $L\left(  x_{0}\right)  =I$ is relaxed to $L\left(
x_{0}\right)  $ commuting with $\tau\left(  w_{0}\right)  $ and being
nonsingular (so $L\left(  xw_{0}^{m}\right)  =\tau\left(  w_{0}\right)
^{-m}L\left(  x\right)  \tau\left(  w_{0}\right)  ^{m}$ still holds for
$x\in\mathcal{C}_{0}$). Then $M\left(  w,x\right)  $ and $L\left(  x\right)  $
have the following properties ($x\in\mathbb{T}_{reg}^{N}$):%
\begin{align}
M\left(  I,x\right)   &  =I;\nonumber\\
M\left(  w_{1}w_{2},x\right)   &  =M\left(  w_{2},xw_{1}\right)  M\left(
w_{1},x\right)  ,~w_{1},w_{2}\in\mathcal{S}_{N},\label{MMMw}\\
L\left(  xw\right)   &  =M\left(  w,x\right)  L\left(  x\right)  \tau\left(
w\right)  ,~w\in\mathcal{S}_{N}\nonumber
\end{align}

With the goal of analyzing vector functions of the form $f\left(  x\right)
=L\left(  x\right)  p\left(  x\right)  $ where $p\in\mathcal{P}_{\tau}$
consider%
\begin{align}
L\left(  x\right)  wp\left(  x\right)   &  =L\left(  x\right)  \tau\left(
w\right)  p\left(  xw\right)  =L\left(  x\right)  \tau\left(  w\right)
L\left(  xw\right)  ^{-1}f\left(  xw\right) \label{LwLi}\\
&  =L\left(  x\right)  \tau\left(  w\right)  \left\{  M\left(  w,x\right)
L\left(  x\right)  \tau\left(  w\right)  \right\}  ^{-1}f\left(  xw\right)
\nonumber\\
&  =M\left(  w,x\right)  ^{-1}f\left(  xw\right)  ,\nonumber
\end{align}
in other words $L\left(  x\right)  wL\left(  x\right)  ^{-1}f\left(  x\right)
=M\left(  w,x\right)  ^{-1}f\left(  xw\right)  $. Accordingly define a twisted
action of $\mathcal{S}_{N}$ on vector-valued functions $f\left(  x\right)  $
(defined on $\mathbb{T}_{reg}^{N}$) by%
\[
\sigma^{M}\left(  w\right)  f\left(  x\right)  =M\left(  w,x\right)
^{-1}f\left(  xw\right)  .
\]

\begin{proposition}
\label{LwLsigma}Suppose $w\in\mathcal{S}_{N}$ then $L\left(  x\right)
wL\left(  x\right)  ^{-1}=\sigma^{M}\left(  w\right)  $ and $\sigma^{M}$ is a
representation of $\mathcal{S}_{N}$.
\end{proposition}

\begin{proof}
Using formula (\ref{MMMw}) let $g\left(  x\right)  =\sigma^{M}\left(
w_{2}\right)  f\left(  x\right)  =M\left(  w_{2},x\right)  ^{-1}f\left(
xw_{2}\right)  $, then%
\begin{align*}
\sigma^{M}\left(  w_{1}\right)  \sigma^{M}\left(  w_{2}\right)  f\left(
x\right)   &  =\sigma^{M}\left(  w_{1}\right)  g\left(  x\right)  =M\left(
w_{1},x\right)  ^{-1}g\left(  xw_{1}\right) \\
&  =M\left(  w_{1},x\right)  ^{-1}M\left(  w_{2},xw_{1}\right)  ^{-1}f\left(
xw_{1}w_{2}\right) \\
&  =M\left(  w_{1}w_{2},x\right)  ^{-1}f\left(  xw_{1}w_{2}\right)
=\sigma^{M}\left(  w_{1}w_{2}\right)  f\left(  x\right)  .
\end{align*}

\end{proof}

If $p\left(  x\right)  $ is symmetric ($\tau\left(  w\right)  p\left(
xw\right)  =p\left(  x\right)  $) then%
\begin{align*}
\sigma^{M}\left(  w\right)  L\left(  x\right)  p\left(  x\right)   &
=M\left(  w,x\right)  ^{-1}L\left(  xw\right)  p\left(  xw\right) \\
&  =M\left(  w,x\right)  ^{-1}\left\{  M\left(  w,x\right)  L\left(  x\right)
\tau\left(  w\right)  \right\}  p\left(  xw\right) \\
&  =L\left(  x\right)  \tau\left(  w\right)  p\left(  xw\right)  =L\left(
x\right)  p\left(  x\right)  .
\end{align*}

This is crucial in the sequel where the operator $L\left(  x\right)
\sum_{i=1}^{N}\left(  \mathcal{U}_{i}-1-\kappa\gamma\right)  ^{2}L\left(
x\right)  ^{-1}$ is related to the Hamiltonian $\mathcal{H}$.

\begin{proposition}
For $1\leq i\leq N$
\begin{equation}
L\left(  x\right)  \left(  \mathcal{U}_{i}-1-\kappa\gamma\right)  L\left(
x\right)  ^{-1}=x_{i}\partial_{i}-\kappa\sum_{j<i}\frac{x_{i}}{x_{i}-x_{j}%
}\sigma^{M}\left(  i,j\right)  -\kappa\sum_{j>i}\frac{x_{j}}{x_{i}-x_{j}%
}\sigma^{M}\left(  i,j\right)  . \label{LULi}%
\end{equation}

\end{proposition}

\begin{proof}
Write the differential system as $\partial_{i}L\left(  x\right)  =\kappa
L\left(  x\right)  A_{i}\left(  x\right)  $ then
\[
0=\partial_{i}\left(  L^{-1}L\right)  =\left(  \partial_{i}L^{-1}\right)
L+L^{-1}\partial_{i}L=\left(  \partial_{i}L^{-1}\right)  L+\kappa L^{-1}%
LA_{i},
\]
thus $\partial_{i}L\left(  x\right)  ^{-1}=-\kappa A_{i}\left(  x\right)
L\left(  x\right)  ^{-1}$. Next by formula (\ref{LwLi}) $L\left(  x\right)
\left(  i,j\right)  L\left(  x\right)  ^{-1}=\sigma^{M}\left(  i,j\right)  $.
For the other term in $\mathcal{U}_{i}=x_{i}\mathcal{D}_{i}+1+\kappa\omega
_{i}$ (formula (\ref{UDx})) we obtain $L\left(  x\right)  \omega_{i}L\left(
x\right)  ^{-1}=\sum_{j>i}\sigma^{M}\left(  i,j\right)  $. Consider%
\begin{align*}
\mathcal{D}_{i}L\left(  x\right)  ^{-1}f\left(  x\right)   &  =\left(
\partial_{i}L\left(  x\right)  ^{-1}\right)  f\left(  x\right)  +L\left(
x\right)  ^{-1}\partial_{i}f\left(  x\right) \\
&  +\kappa\sum_{j\neq i}\frac{\tau\left(  i,j\right)  }{x_{i}-x_{j}}\left\{
L\left(  x\right)  ^{-1}f\left(  x\right)  -L\left(  x\left(  i,j\right)
\right)  ^{-1}f\left(  x\left(  i,j\right)  \right)  \right\} \\
&  =-\kappa\left\{  \sum_{j\neq i}\frac{\tau\left(  i,j\right)  }{x_{i}-x_{j}%
}-\frac{\gamma}{x_{i}}I\right\}  L\left(  x\right)  ^{-1}f\left(  x\right)
+L\left(  x\right)  ^{-1}\partial_{i}f\left(  x\right) \\
&  +\kappa\sum_{j\neq i}\frac{\tau\left(  i,j\right)  }{x_{i}-x_{j}}\left\{
L\left(  x\right)  ^{-1}f\left(  x\right)  -\tau\left(  i,j\right)
^{-1}L\left(  x\right)  ^{-1}M\left(  \left(  i,j\right)  ,x\right)
^{-1}f\left(  x\left(  i,j\right)  \right)  \right\} \\
&  =L\left(  x\right)  ^{-1}\left\{  \frac{\kappa\gamma}{x_{i}}f\left(
x\right)  +\partial_{i}f\left(  x\right)  -\kappa\sum_{j\neq i}\frac{1}%
{x_{i}-x_{j}}\sigma^{M}\left(  i,j\right)  f\left(  x\right)  .\right\}  .
\end{align*}
Thus
\begin{align*}
&  L\left(  x\right)  \left\{  x_{i}\mathcal{D}_{i}+1+\kappa\omega
_{i}\right\}  L\left(  x\right)  ^{-1}-1-\kappa\gamma\\
&  =x_{i}\partial_{i}-\kappa\sum_{j\neq i}\frac{x_{i}}{x_{i}-x_{j}}\sigma
^{M}\left(  i,j\right)  +\kappa\sum_{j>i}\sigma^{M}\left(  i,j\right) \\
&  =x_{i}\partial_{i}-\kappa\sum_{j<i}\frac{x_{i}}{x_{i}-x_{j}}\sigma
^{M}\left(  i,j\right)  -\kappa\sum_{j>i}\frac{x_{j}}{x_{i}-x_{j}}\sigma
^{M}\left(  i,j\right)  .
\end{align*}

\end{proof}

We will use an elementary double sum formula: suppose $g\left(  i,j\right)  $
is a function defined on all pairs $\left(  i,j\right)  $ with $1\leq i,j\leq
N$ then%
\begin{equation}
\sum_{i,j=1,i\neq j}^{N}g\left(  i,j\right)  =\sum_{1\leq i<j\leq N}\left\{
g\left(  i,j\right)  +g\left(  j,i\right)  \right\}  . \label{dblsumg1}%
\end{equation}

\begin{corollary}
$\sum_{i=1}^{N}\mathcal{U}_{i}=\sum_{i=1}^{N}x_{i}\partial_{i}+N+\kappa
S_{1}.$
\end{corollary}

\begin{proof}
From formula (\ref{LULi})%
\[
L\left(  x\right)  \sum_{i=1}^{N}\left(  \mathcal{U}_{i}-1-\kappa
\gamma\right)  L\left(  x\right)  ^{-1}=\sum_{i=1}^{N}x_{i}\partial_{i}%
-\kappa\sum_{i=1}^{N}\sum_{j\neq i}\frac{x_{\max\left(  i,j\right)  }}%
{x_{i}-x_{j}}\sigma^{M}\left(  i,j\right)  .
\]
The double sum is of the form (\ref{dblsumg1}) and for $i<j$ one obtains
$g\left(  i,j\right)  +g\left(  j,i\right)  =\frac{x_{i}}{x_{i}-x_{j}}%
\sigma^{M}\left(  i,j\right)  +\frac{x_{i}}{x_{j}-x_{i}}\sigma^{M}\left(
j,i\right)  =0$. From $\sum_{i=1}^{N}x_{i}\partial_{i}L\left(  x\right)  =0$
it follows that $\sum_{i=1}^{N}x_{i}\partial_{i}$ commutes with $L\left(
x\right)  $ and together with $\gamma=S_{1}\left(  \tau\right)  /N$ completes
the proof. to get the stated formula.
\end{proof}

\begin{lemma}
For $1\leq i\leq N$%
\begin{gather}
L\left(  x\right)  \left(  \mathcal{U}_{i}-1-\kappa\gamma\right)  ^{2}L\left(
x\right)  ^{-1}=\left(  x_{i}\partial_{i}\right)  ^{2}-\kappa\sum_{j\neq
i}\frac{x_{i}x_{j}}{\left(  x_{i}-x_{j}\right)  ^{2}}\left(  \kappa-\sigma
^{M}\left(  i,j\right)  \right) \label{LU2L1}\\
-\kappa\sum_{j<i}\left\{  \frac{x_{i}^{2}}{x_{i}-x_{j}}\sigma^{M}\left(
i,j\right)  \partial_{j}f+\frac{x_{j}x_{i}}{x_{i}-x_{j}}\sigma^{M}\left(
i,j\right)  \partial_{i}f\right\} \label{LU2L2}\\
-\kappa\sum_{j>i}\left\{  \frac{x_{j}x_{i}}{x_{i}-x_{j}}\sigma^{M}\left(
i,j\right)  \partial_{j}f+\frac{x_{i}^{2}}{x_{i}-x_{j}}\sigma^{M}\left(
i,j\right)  \partial_{i}f\right\} \label{LU2L3}\\
+\kappa^{2}\sum_{\#\left\{  i,j,k\right\}  =3}\frac{x_{\max\left(  i,j\right)
}}{x_{i}-x_{j}}\sigma^{M}\left(  i,j\right)  \frac{x_{\max\left(  i,k\right)
}}{x_{i}-x_{k}}\sigma^{M}\left(  i,k\right)  . \label{LU2L4}%
\end{gather}

\end{lemma}

\begin{proof}
In the square of formula (\ref{LULi}) group the terms as%
\begin{align*}
&  -\kappa x_{i}\partial_{i}\sum_{j\neq i}\frac{x_{\max\left(  i,j\right)  }%
}{x_{i}-x_{j}}\sigma^{M}\left(  i,j\right)  -\kappa\sum_{j\neq i}\frac
{x_{\max\left(  i,j\right)  }}{x_{i}-x_{j}}\sigma^{M}\left(  i,j\right)
x_{i}\partial_{i}\\
&  =\kappa\sum_{j\neq i}\frac{x_{i}x_{j}}{\left(  x_{i}-x_{j}\right)  ^{2}%
}\sigma^{M}\left(  i,j\right)  -\kappa\sum_{j\neq i}\frac{x_{\max\left(
i,j\right)  }}{x_{i}-x_{j}}\left\{  x_{i}\sigma^{M}\left(  i,j\right)
\partial_{j}+x_{j}\sigma^{M}\left(  i,j\right)  \partial_{i}\right\}  ,
\end{align*}
because%
\begin{align*}
\partial_{i}\left(  \sigma^{M}\left(  i,j\right)  f\right)  \left(  x\right)
&  =M\left(  \left(  i,j\right)  ,x\right)  ^{-1}\partial_{i}f\left(  x\left(
i,j\right)  \right)  =M\left(  \left(  i,j\right)  ,x\right)  ^{-1}\left(
\partial_{j}f\right)  \left(  x\left(  i,j\right)  \right) \\
&  =\sigma^{M}\left(  i,j\right)  \left(  \partial_{j}f\right)  \left(
x\right)  ,
\end{align*}
and $M\left(  w,x\right)  $ is locally constant in $x$. Next consider%
\begin{align*}
&  \kappa^{2}\sum_{j<i}\frac{x_{i}}{x_{i}-x_{j}}\sigma^{M}\left(  i,j\right)
\frac{x_{i}}{x_{i}-x_{j}}\sigma^{M}\left(  i,j\right)  +\kappa^{2}\sum
_{j>i}\frac{x_{j}}{x_{i}-x_{j}}\sigma^{M}\left(  i,j\right)  \frac{x_{j}%
}{x_{i}-x_{j}}\sigma^{M}\left(  i,j\right) \\
&  =-\kappa^{2}\sum_{j<i}\frac{x_{i}x_{j}}{\left(  x_{i}-x_{j}\right)  }%
\sigma^{M}\left(  i,j\right)  ^{2}-\kappa^{2}\sum_{j>i}\frac{x_{i}x_{j}%
}{\left(  x_{i}-x_{j}\right)  ^{2}}\sigma^{M}\left(  i,j\right)  ^{2}%
=-\kappa^{2}\sum_{j>i}\frac{x_{i}x_{j}}{\left(  x_{i}-x_{j}\right)  ^{2}},
\end{align*}
because $\sigma^{M}\left(  i,j\right)  \dfrac{x_{i}}{x_{i}-x_{j}}=\dfrac
{x_{j}}{x_{j}-x_{i}}\sigma^{M}\left(  i,j\right)  $ and $\sigma^{M}\left(
i,j\right)  ^{2}=I$ (by Proposition \ref{LwLsigma}).
\end{proof}

More detailed analysis of the terms in line (\ref{LU2L4}) shows that there are
four different coefficients of $\dfrac{\kappa^{2}}{\left(  x_{i}-x_{j}\right)
\left(  x_{j}-x_{k}\right)  }\sigma^{M}\left(  \left(  i,j\right)  \left(
i,k\right)  \right)  $ depending on the numerical order of $i,j,k$:

\begin{enumerate}
\item $x_{i}x_{j}$ if $j<k<i$ or $k<j<i$,

\item $x_{i}x_{k}$ if $j<i<k$,

\item $x_{j}x_{k}$ if $i<j<k$ or $i<k<j$,

\item $x_{j}^{2}$ if $k<i<j$.
\end{enumerate}

The next step is to sum over $1\leq i\leq N$. Lines (\ref{LU2L2},\ref{LU2L3})
sum to zero by using Formula (\ref{dblsumg1}). We show that all the terms in
line (\ref{LU2L4}) sum to zero. This is a sum over all cycles of order $3$.
Any $3$-cycle is of the form $\left(  a,b,c\right)  $ with $1\leq a<b<c\leq N$
or $1\leq a<c<b\leq N$. Each $3$-cycle appears three times in the sum since
\[
\left(  a,c\right)  \left(  a,b\right)  =\left(  b,a\right)  \left(
b,c\right)  =\left(  c,b\right)  \left(  c,a\right)  =\left(  a,b,c\right)  .
\]
If $a<b<c$ then by the above formulae the coefficient of $\kappa^{2}\sigma
^{M}\left(  \left(  a,b,c\right)  \right)  $ is
\[
\frac{x_{b}x_{c}}{\left(  x_{a}-x_{c}\right)  \left(  x_{c}-x_{b}\right)
}+\frac{x_{b}x_{c}}{\left(  x_{b}-x_{a}\right)  \left(  x_{a}-x_{c}\right)
}+\frac{x_{b}x_{c}}{\left(  x_{c}-x_{b}\right)  \left(  x_{b}-x_{a}\right)
}=0.
\]
If $a<c<b$ then the coefficient of $\kappa^{2}\sigma^{M}\left(  \left(
a,b,c\right)  \right)  $ is%
\[
\frac{x_{b}x_{c}}{\left(  x_{a}-x_{c}\right)  \left(  x_{c}-x_{b}\right)
}+\frac{x_{b}x_{a}}{\left(  x_{b}-x_{a}\right)  \left(  x_{a}-x_{c}\right)
}+\frac{x_{b}^{2}}{\left(  x_{c}-x_{b}\right)  \left(  x_{b}-x_{a}\right)
}=0.
\]
Each pair $\left\{  i,j\right\}  $ appears twice in the sum of the terms in
(\ref{LU2L1}). We have proven the following:

\begin{theorem}
\label{LHL}The $L\left(  x\right)  $ conjugate of $\sum_{i=1}^{N}\left(
\mathcal{U}_{i}-1-\kappa\gamma\right)  ^{2}$ is
\begin{align*}
\mathcal{H}_{M}  &  :=L\left(  x\right)  \sum_{i=1}^{N}\left(  \mathcal{U}%
_{i}-1-\kappa\gamma\right)  ^{2}L\left(  x\right)  ^{-1}\\
&  =\sum_{i=1}^{N}\left(  x_{i}\partial_{i}\right)  ^{2}-2\kappa\sum_{1\leq
i<j\leq N}\frac{x_{i}x_{j}}{\left(  x_{i}-x_{j}\right)  ^{2}}\left(
\kappa-\sigma^{M}\left(  i,j\right)  \right)  .
\end{align*}

\end{theorem}

Thus $\mathcal{H}_{M}$ agrees with $\mathcal{H}$ when applied to $L\left(
x\right)  p\left(  x\right)  $ where $p$ is a symmetric ($\tau\left(
w\right)  p\left(  xw\right)  =p\left(  x\right)  )$ polynomial. By
Proposition \ref{Usym} the operators $L\left(  x\right)  \sum_{i=1}%
^{N}\mathcal{U}_{i}^{m}L\left(  x\right)  ^{-1}$ commute with $\mathcal{H}%
_{M}$ and with $\sigma^{M}\left(  w\right)  $ for $w\in\mathcal{S}_{N}$, for
$m=1,2,3,\ldots$ and $L\left(  x\right)  p\left(  x\right)  $ is an
eigenfunction of $\mathcal{H}_{M}$ for any NSJP $p\left(  x\right)  $.

\section{\label{Hforms}Hermitian forms and nonsymmetric Jack polynomials}

The results in this section come from \cite{D2017},\cite{DL2011},\cite{G2010}.
To obtain square-integrable and mutually orthogonal wavefunctions we start
with a Hermitian form $\left\langle \cdot,\cdot\right\rangle _{\mathbb{T}}$
for $\mathcal{P}_{\tau}$ with the properties ($f,g\in\mathcal{P}_{\tau};1\leq
i\leq N;c\in\mathbb{C};T,T^{\prime}\in\mathcal{Y}\left(  \tau\right)  $)%
\begin{align}
\left\langle 1\otimes T,1\otimes T^{\prime}\right\rangle _{\mathbb{T}}  &
=\left\langle T,T^{\prime}\right\rangle _{0},\label{Hform}\\
\left\langle f,g\right\rangle _{\mathbb{T}}  &  =\overline{\left\langle
g,f\right\rangle _{\mathbb{T}}},~\left\langle f,cg\right\rangle _{\mathbb{T}%
}=c\left\langle f,g\right\rangle _{\mathbb{T}},\nonumber\\
\left\langle wf,wg\right\rangle _{\mathbb{T}}  &  =\left\langle
f,g\right\rangle _{\mathbb{T}},~w\in\mathcal{S}_{N},\nonumber\\
\left\langle x_{i}\mathcal{D}_{i}f,g\right\rangle _{\mathbb{T}}  &
=\left\langle f,x_{i}\mathcal{D}_{i}g\right\rangle _{\mathbb{T}},\nonumber\\
\left\langle x_{i}f,x_{i}g\right\rangle _{\mathbb{T}}  &  =\left\langle
f,g\right\rangle _{\mathbb{T}}.\nonumber
\end{align}
The properties define the form uniquely and imply $\left\langle \mathcal{U}%
_{i}f,g\right\rangle _{\mathbb{T}}=\left\langle f,\mathcal{U}_{i}%
g\right\rangle _{\mathbb{T}}$ and thus the orthogonality of the NSJP's
(Theorem \ref{NSJPprop} below). The form is not defined for all $\kappa$ and
need not be positive-definite. The key results from \cite{D2017} (recall the
maximum hook-length $h_{\tau}$ from (\ref{hllt})) are:

\begin{theorem}
Suppose $-1/h_{\tau}<\kappa<1/h_{\tau}$ and $L_{0}\left(  x\right)  $ is the
solution of (\ref{Lsys}) satisfying $L_{0}\left(  x_{0}\right)  =I$ and
extended to $\mathbb{T}_{reg}^{N}$ by (\ref{TregL}) then there exists a unique
positive-definite matrix $B$ such that $B\tau\left(  w_{0}\right)
=\tau\left(  w_{0}\right)  B$ and%
\[
\left\langle f,g\right\rangle _{\mathbb{T}}=\int_{\mathbb{T}^{N}}f\left(
x\right)  ^{\ast}L_{0}\left(  x\right)  ^{\ast}BL_{0}\left(  x\right)
g\left(  x\right)  \mathrm{d}m\left(  x\right)  .
\]

\end{theorem}

Each $f\in\mathcal{P}_{\tau}$ has the expansion $\sum_{T\in\mathcal{Y}\left(
\tau\right)  }\left\langle T,T\right\rangle ^{-1/2}f_{T}\left(  x\right)
\otimes T$ with $f_{T}\in\mathcal{P}$ and $f\left(  x\right)  $ is considered
as a column vector $\left[  f_{T}\right]  _{T\in\mathcal{Y}\left(
\tau\right)  }$ in the integral formula. It is implicit in the theorem that
$L_{0}\left(  x\right)  ^{\ast}BL_{0}\left(  x\right)  $ is integrable, and
$B$ depends on $\kappa$. Henceforth we use $\left\Vert p\right\Vert
^{2}:=\left\langle p,p\right\rangle _{\mathbb{T}}$ (which need not be positive
for $\kappa$ outside the above interval).

There is a unique positive-definite matrix $C$ such that $C^{2}=B$; as a
consequence $C$ commutes with $\tau\left(  w_{0}\right)  $ (because there is
real polynomial $r\left(  t\right)  $ such that $r\left(  B\right)  =C$). We
apply the results of the previous section to%
\begin{equation}
L\left(  x\right)  :=CL_{0}\left(  x\right)  , \label{goodL}%
\end{equation}
and the integral formula becomes%
\[
\int_{\mathbb{T}^{N}}\left\{  L\left(  x\right)  f\left(  x\right)  \right\}
^{\ast}L\left(  x\right)  g\left(  x\right)  \mathrm{d}m\left(  x\right)
=\left\langle f,g\right\rangle _{\mathbb{T}}.
\]
Here is an outline of the structure and properties of NSJP's: The operators
$\mathcal{U}_{i}$ have a triangularity property with respect to a partial
order on $\mathbb{N}_{0}^{N}$. For $\alpha\in\mathbb{N}_{0}^{N}$ let
$\alpha^{+}$ denote the nonincreasing rearrangement of $\alpha$ so that
$\alpha^{+}$ is a partition.

\begin{definition}
The dominance order $\prec$ and the derived order $\vartriangleleft$ on
$\mathbb{N}_{0}^{N}$ are given by (i) $\alpha\prec\beta$ if and only if
$\sum_{j=1}^{i}\alpha_{j}\leq\sum_{j=1}^{i}\beta_{j},$ for$~1\leq i\leq N$
and$~\alpha\neq\beta$; (ii) $\alpha\vartriangleleft\beta$ if and only if
$\left\vert \alpha\right\vert =\left\vert \beta\right\vert $, $\alpha^{+}%
\prec\beta^{+}$ ,or $\alpha^{+}=\beta^{+}$ and $\alpha\prec\beta$.
\end{definition}

For example: $\left(  3,1,1\right)  \vartriangleleft\left(  0,2,4\right)
\vartriangleleft$ $\left(  4,0,2\right)  $; while $\left(  4,1,1\right)
,\left(  3,3,0\right)  $ are not $\vartriangleleft$-comparable.

The NSJP's are labeled by pairs $\left(  \alpha,T\right)  \in\mathbb{N}%
_{0}^{N}\times\mathcal{Y}\left(  \tau\right)  $ but the leading term involves
a twist.

\begin{definition}
For $\alpha\in\mathbb{N}_{0}^{N}$ the rank function on $\left\{
1,\ldots,N\right\}  $ is given by%
\[
r_{\alpha}\left(  i\right)  =\#\left\{  j:\alpha_{j}>\alpha_{i}\right\}
+\#\left\{  j:1\leq j\leq i,\alpha_{j}=\alpha_{i}\right\}  ,
\]
then $r_{\alpha}\in\mathcal{S}_{N}$ and $r_{\alpha}\alpha=\alpha^{+}$ the
nonincreasing rearrangement of $\alpha.$
\end{definition}

For example if $\alpha=\left(  1,2,1,4\right)  $ then $r_{\alpha}=\left[
3,2,4,1\right]  $ and $r_{\alpha}\alpha=\alpha^{+}=\left(  4,2,1,1\right)  $
(recall $w\alpha_{i}=\alpha_{w^{-1}\left(  i\right)  }$ ).

\begin{theorem}
\label{NSJPprop}For $\left(  \alpha,T\right)  \in\mathbb{N}_{0}^{N}%
\times\mathcal{Y}\left(  \tau\right)  $ and for all $\kappa$ except for a
discrete subset of $\mathbb{Q}$ there is a unique simultaneous eigenfunction
$\zeta_{\alpha,T}\in\mathcal{P}_{\tau}$ of $\left\{  \mathcal{U}_{i}\right\}
$, homogeneous of degree $\left\vert \alpha\right\vert $, such that%
\begin{align*}
\zeta_{\alpha,T}  &  =x^{\alpha}\otimes\tau\left(  r_{\alpha}^{-1}\right)
T+\sum_{\beta\vartriangleleft\alpha}x^{\beta}\otimes t_{\alpha\beta}\left(
\kappa\right)  ,t_{\alpha\beta}\left(  \kappa\right)  \in V_{\tau},\\
\mathcal{U}_{i}\zeta_{\alpha,T}  &  =\left(  \alpha_{i}+1+\kappa c\left(
r_{\alpha}\left(  i\right)  ,T\right)  \right)  \zeta_{\alpha,T},~1\leq i\leq
N.
\end{align*}

\end{theorem}

The $\zeta_{\alpha,T}$ are called nonsymmetric Jack polynomials. The condition
on $\kappa$ for existence is satisfied if each pair $\left(  \alpha,T\right)
$ is determined by its \textit{spectral vector} $\xi_{\alpha,T}:=\left[
\alpha_{i}+1+\kappa c\left(  r_{\alpha}\left(  i\right)  ,T\right)  \right]
_{i=1}^{N}$; this includes the interval $-1/h_{\tau}\leq\kappa\leq1/h_{\tau}$.
There is an algorithmic approach to the construction based on the Yang-Baxter
(directed) graph. The edges involve the adjacent transpositions $s_{i}$, which
act by transposition on the spectral vector, and a degree-raising operation
which shifts and increments the spectral vector. The nodes of the graph are of
the form%
\[
\left(  \alpha,T,\xi_{\alpha,T},r_{\alpha},\zeta_{\alpha,T}\right)  ,
\]
(abbreviated to $\left(  \alpha,T\right)  $) the root is $\left(  0^{N}%
,T_{0},\left[  1+\kappa c\left(  i,T_{0}\right)  \right]  _{i=1}%
^{N},I,1\otimes T_{0}\right)  $ where $T_{0}$ is formed by entering
$N,N-1,\ldots,1$ column-by-column in the Ferrers diagram. The degree-raising
edge uses the map $\Phi:\left(  c_{1},c_{2},\ldots,c_{N}\right)
\rightarrow\left(  c_{2},c_{3},\ldots,c_{N},c_{1}+1\right)  $ on $N$-tuples.
It is called an\textit{ affine step }and is defined by%
\[
\left(  \alpha,T,\xi_{\alpha,T},\alpha,r_{\alpha},\zeta_{\alpha,T}\right)
\overset{\Phi}{\longrightarrow}\left(  \Phi\alpha,T,\Phi\xi_{\alpha
,T},r_{\alpha}w_{0},\zeta_{\Phi\alpha,T}\right)  ,
\]%
\[
\zeta_{\Phi\alpha,T}=x_{N}w_{0}^{-1}\zeta_{\alpha,T};
\]
(recall $w_{0}=\left(  1,2,\ldots,N\right)  $, an $N$-cycle) the leading term
is $x^{\Phi\alpha}\otimes\tau\left(  w_{0}^{-1}r_{\alpha}^{-1}\right)  T$ and
$w_{0}^{-1}r_{\alpha}^{-1}=r_{\Phi\alpha}^{-1}$ because $r_{\Phi\alpha
}=r_{\alpha}w_{0}$ for any $\alpha$: $r_{\alpha}w_{0}\left(  i\right)
=r_{\alpha}\left(  w_{0}\left(  i\right)  \right)  =r_{\alpha}\left(
i+1\right)  $ for $1\leq i<N$, $r_{\alpha}w_{0}\left(  N\right)  =r_{\alpha
}\left(  1\right)  $. For example: $\alpha=\left(  0,3,5,0\right)  $,
$r_{\alpha}=\left[  3,2,1,4\right]  $; $\Phi\alpha=\left(  3,5,0,1\right)  $,
$r_{\Phi\alpha}=\left[  2,1,4,3\right]  $.

The other edges are called steps or jumps, both labeled by $s_{i}$: the
formulae for both rely on the commutation (\ref{commU}) and the coefficient
$b$ is determined by the condition that $s_{i}\zeta_{\alpha,T}-b\zeta
_{\alpha,T}$ is an eigenfunction of $\mathcal{U}_{i}$.

If $\alpha_{i}<\alpha_{i+1}$, then the step $s_{i}$ is%
\[
\left(  \alpha,T,\xi_{\alpha,T},r_{\alpha},\zeta_{\alpha,T}\right)
\overset{s_{i}}{\longrightarrow}\left(  s_{i}\alpha,T,s_{i}\xi_{\alpha
,T},r_{\alpha}s_{i},\zeta_{s_{i}\alpha,T}\right)
\]%
\[
\zeta_{s_{i}\alpha,T}=s_{i}\zeta_{\alpha,T}-\frac{\kappa}{\xi_{\alpha
,T}\left(  i\right)  -\xi_{\alpha,T}\left(  i+1\right)  }\zeta_{\alpha,T}.
\]

If $\alpha_{i}=\alpha_{i+1}$, set $j=r_{\alpha}\left(  i\right)  $, so that
$j+1=r_{\alpha}\left(  i+1\right)  $ and $s_{i}r_{\alpha}^{-1}=r_{\alpha}%
^{-1}s_{j}$. Thus $\xi_{\alpha,T}\left(  i\right)  =\alpha_{i}+1+\kappa
c\left(  j,T\right)  $ and $\xi_{\alpha,T}\left(  i+1\right)  =\alpha
_{i}+1+\kappa c\left(  j+1,T\right)  $. Set%
\[
b^{\prime}=\frac{1}{c\left(  j,T\right)  -c\left(  j+1,T\right)  };
\]
If $b^{\prime}=1$ \{$\mathrm{rw}\left(  j,T\right)  =\mathrm{rw}\left(
j+1,T\right)  $\} or $-1$ \{$\mathrm{cm}\left(  j,T\right)  =\mathrm{cm}%
\left(  j+1,T\right)  $\} then $s_{i}\zeta_{\alpha,T}=\zeta_{\alpha,T}$ or
$-\zeta_{\alpha,T}$ respectively. Otherwise let $T^{\left(  j\right)  }$
denote the result of interchanging $j$ and $j+1$ in $T$. If $0<b^{\prime}%
\leq\frac{1}{2}$, that is, $\mathrm{rw}\left(  j,T\right)  <\mathrm{rw}\left(
j+1,T\right)  $ (and $\mathrm{cm}\left(  j,T\right)  >\mathrm{cm}\left(
j+1,T\right)  $; if one takes the $\left(  1,1\right)  $ cell of $T$ as
northwest then $j$ is northeast of $j+1$) then the jump $s_{i}$ is
(\textquotedblleft jump\textquotedblright\ suggests jumping from one tableau
to another)
\[
\left(  \alpha,T,\xi_{\alpha,T},r_{\alpha},\zeta_{\alpha,T}\right)
\overset{s_{i}}{\longrightarrow}\left(  \alpha,T^{\left(  j\right)  },s_{i}%
\xi_{\alpha,T},r_{\alpha},\zeta_{\alpha,T^{\left(  j\right)  }}\right)  ,
\]%
\begin{equation}
\zeta_{\alpha,T^{\left(  j\right)  }}=s_{i}\zeta_{\alpha,T}-b^{\prime}%
\zeta_{\alpha,T}, \label{jumpT}%
\end{equation}
The leading term is transformed $s_{i}\left(  x^{\alpha}\otimes\tau\left(
r_{\alpha}^{-1}\right)  T\right)  =\left(  xs_{i}\right)  ^{\alpha}\otimes
\tau\left(  s_{i}r_{\alpha}^{-1}\right)  T=x^{\alpha}\otimes\tau\left(
r_{\alpha}^{-1}\right)  \tau\left(  s_{j}\right)  T$ and $\tau\left(
s_{j}\right)  T=T^{\left(  j\right)  }+b^{\prime}T$. The jump applies to the
situation $\alpha=0^{N}$ and provides the transformation formulae for $s_{i}$
acting on $T\in\mathcal{Y}\left(  \tau\right)  $ (that is, on $1\otimes T$ and
$r_{\boldsymbol{\alpha}}\left(  i\right)  =i$).

\begin{example}
\label{ex(2,1)}Let $N=3,\tau=\left(  2,1\right)  $ and $T_{0}=%
\begin{array}
[c]{cc}%
3 & 1\\
2 &
\end{array}
,T_{1}=%
\begin{array}
[c]{cc}%
3 & 2\\
1 &
\end{array}
$, and consider $\left(  \alpha,T\right)  =\left(  \left(  0,1,1\right)
,T_{0}\right)  .$ Then $r_{\alpha}=\left[  3,1,2\right]  $ and $\xi
_{\alpha,T_{0}}=\left[  1,2+\kappa,2-\kappa\right]  $. The step $s_{1}$ is
$\zeta_{\left(  1,0,1\right)  ,T_{0}}=s_{1}\zeta_{\alpha,T_{0}}+\dfrac{\kappa
}{1+\kappa}\zeta_{\alpha,T_{0}}$; for the jump $s_{2}$ one finds
$j=1,b^{\prime}=\frac{1}{2}$ and $\zeta_{\left(  0,1,1\right)  ,T_{1}}%
=s_{2}\zeta_{\alpha,T_{0}}-\frac{1}{2}\zeta_{\alpha,T_{0}}$. Note
$\xi_{\left(  0,1,1\right)  ,T_{1}}=\left[  1,2-\kappa,2+\kappa\right]  $.
\end{example}

The hypotheses on the Hermitian form (\ref{Hform}) imply%
\[
\xi_{\alpha,T}\left(  i\right)  \left\langle \zeta_{\alpha,T},\zeta
_{\beta,T^{\prime}}\right\rangle _{\mathbb{T}}=\left\langle \mathcal{U}%
_{i}\zeta_{\alpha,T},\zeta_{\beta,T^{\prime}}\right\rangle _{\mathbb{T}%
}=\left\langle \zeta_{\alpha,T},\mathcal{U}_{i}\zeta_{\beta,T^{\prime}%
}\right\rangle _{\mathbb{T}}=\xi_{\beta,T^{\prime}}\left(  i\right)
\left\langle \zeta_{\alpha,T},\zeta_{\beta,T^{\prime}}\right\rangle
_{\mathbb{T}}%
\]
and thus $\left\langle \zeta_{\alpha,T},\zeta_{\beta,T^{\prime}}\right\rangle
_{\mathbb{T}}=0$, (for permitted values of $\kappa$). The orthogonality
provides an inductive process for computing $\left\langle \zeta_{\alpha
,T},\zeta_{\alpha,T}\right\rangle _{\mathbb{T}}$: for the step $s_{i}$ with
$\alpha_{i}<\alpha_{i+1}$ we have $s_{i}\zeta_{\alpha,T}=\zeta_{s_{i}\alpha
,T}+b\zeta_{\alpha,T}$ (where $b=\frac{\kappa}{\xi_{\alpha,T}\left(  i\right)
-\xi_{\alpha,T}\left(  i+1\right)  }$) and%
\begin{align*}
\left\Vert \zeta_{\alpha,T}\right\Vert ^{2}  &  =\left\Vert s_{i}\zeta
_{\alpha,T}\right\Vert ^{2}=\left\Vert \zeta_{s_{i}\alpha,T}\right\Vert
^{2}+b^{2}\left\Vert \zeta_{\alpha,T}\right\Vert ^{2},\\
\left\Vert \zeta_{s_{i}\alpha,T}\right\Vert ^{2}  &  =\left(  1-b^{2}\right)
\left\Vert \zeta_{\alpha,T}\right\Vert ^{2}.
\end{align*}
A similar formula holds for the jump ($\alpha_{i}=\alpha_{i+1}$). For the
affine step, the hypotheses (\ref{Hform}) imply $\left\Vert \zeta_{\Phi
\alpha,T}\right\Vert ^{2}=\left\Vert \zeta_{\alpha,T}\right\Vert ^{2}$.
Together with $\left\langle 1\otimes T,1\otimes T^{\prime}\right\rangle
_{\mathbb{T}}=\left\langle T,T^{\prime}\right\rangle _{0}$ this procedure
leads to formulae for all $\left\Vert \zeta_{\alpha,T}\right\Vert ^{2}$.

\begin{theorem}
\label{Pnorm1}For $\lambda\in\mathbb{N}_{0}^{N,+}$ ($\lambda_{1}\geq
\lambda_{2}\ldots\geq\lambda_{N}$) and $T\in\mathcal{Y}\left(  \tau\right)  $%
\[
\left\Vert \zeta_{\lambda,T}\right\Vert ^{2}=\left\langle T,T\right\rangle
_{0}\prod_{1\leq i<j\leq N}\prod_{\ell=1}^{\lambda_{i}-\lambda_{j}}\left(
1-\left(  \frac{\kappa}{\ell+\kappa\left(  c\left(  i,T\right)  -c\left(
j,T\right)  \right)  }\right)  ^{2}\right)  .
\]
\newline
\end{theorem}

There is an additional factor for nonpartition indices.

\begin{definition}
For $\alpha\in\mathbb{N}_{0}^{N},T\in\mathcal{Y}\left(  \tau\right)  $ and
$\varepsilon=\pm1$ set%
\[
\mathcal{E}_{\varepsilon}\left(  \alpha,T\right)  :=\prod_{\substack{1\leq
i<j\leq N\\\alpha_{i}<\alpha_{j}}}\left(  1+\frac{\varepsilon\kappa}%
{\alpha_{j}-\alpha_{i}+\kappa\left(  c\left(  r_{\alpha}\left(  j\right)
,T\right)  -c\left(  r_{\alpha}\left(  i\right)  ,T\right)  \right)  }\right)
.
\]

\end{definition}

\begin{theorem}
Suppose $\alpha\in\mathbb{N}_{0}^{N},T\in\mathcal{Y}\left(  \tau\right)  $
then $\left\Vert \zeta_{\alpha,T}\right\Vert ^{2}=\left(  \mathcal{E}%
_{1}\left(  \alpha,T\right)  \mathcal{E}_{-1}\left(  \alpha,T\right)  \right)
^{-1}\left\Vert \zeta_{\alpha^{+},T}\right\Vert ^{2}$.
\end{theorem}

It is important that $-1/h_{\tau}<\kappa<1/h_{\tau}$ implies $\left\Vert
\zeta_{\alpha,T}\right\Vert ^{2}>0$ for all $\left(  \alpha,T\right)  $ and
thus $\left\langle \cdot,\cdot\right\rangle _{\mathbb{T}}$ is
positive-definite. Observe that the value of $\left\Vert \zeta_{\alpha
,T}\right\Vert ^{2}$ depends only on the differences $\alpha_{i}-\alpha_{j}$.
This is a consequence of the torus property $\left\langle x_{i}f,x_{i}%
g\right\rangle _{\mathbb{T}}=\left\langle f,g\right\rangle _{\mathbb{T}}$ and
the commutation (where $e_{N}:=x_{1}x_{2}\cdots x_{N}$ )%
\[
\mathcal{U}_{i}\left(  e_{N}^{m}f\right)  =me_{N}^{m}f+e_{N}^{m}%
\mathcal{U}_{i}f,~1\leq i\leq N,m\in\mathbb{N},
\]
Thus $\mathcal{U}_{i}\left(  e_{N}^{m}\zeta_{\alpha,T}\right)  =\left(
m+\alpha_{i}+\kappa c\left(  r_{\alpha}\left(  i\right)  ,T\right)  \right)
e_{N}^{m}\zeta_{\alpha,T}$, and $e_{N}^{m}\zeta_{\alpha,T}$ is a simultaneous
eigenfunction of $\left\{  \mathcal{U}_{i}\right\}  $ with the same
eigenvalues and the same leading term as $\zeta_{\alpha+m\boldsymbol{1},T}$
for $m\geq0$ (with $\boldsymbol{1}:=\left(  1,1,\ldots,1\right)  \in
\mathbb{N}_{0}^{N}$). Hence $\zeta_{\alpha+m\boldsymbol{1},T}=e_{N}^{m}%
\zeta_{\alpha,T}$. There are Laurent polynomial eigenfunctions of $\left\{
\mathcal{U}_{i}\right\}  $. The structure of NSJP's is extended to $V_{\tau}%
$-valued Laurent polynomials, thereby producing a basis:

\begin{definition}
\label{eNmult}Suppose $\alpha\in\mathbb{Z}^{N}$ then set $\zeta_{\alpha
,T}=e_{N}^{-m}\zeta_{\alpha+m\boldsymbol{1},T}$ where $m\in\mathbb{N}_{0}$ and
satisfies $m\geq-\min_{j}\alpha_{i}$. This is well-defined since
$\alpha+m\boldsymbol{1}\in\mathbb{N}_{0}^{N}$ and $\zeta_{\alpha
+k\boldsymbol{1},T}=e_{N}^{k}\zeta_{\alpha,T}$ for $k\in\mathbb{N}_{0}$.
\end{definition}

\section{\label{symmvp}Symmetric vector-valued polynomials}

For an arbitrary $\left(  \alpha,T\right)  \in\mathbb{N}_{0}^{N}%
\times\mathcal{Y}\left(  \tau\right)  $ we can define a symmetric polynomial
simply by averaging: $p=\frac{1}{N!}\sum_{w\in\mathcal{S}_{N}}w\zeta
_{\alpha,T}$. From Proposition \ref{Usym} it follows that $p$ is an
eigenfunction of $\sum_{i=1}^{N}\mathcal{U}_{i}^{m}$ for each $m=1,2,3,\ldots
$. The idea of using this method to construct Jack polynomials from the scalar
nonsymmetric Jack polynomials is due to Baker and Forrester \cite{B1999}; the
usual Jack parameter is $\alpha=1/\kappa$. It is possible that for some
$\left(  \alpha,T\right)  $ the sum $p=0$, and for some pairs $\left(
\alpha,T\right)  $ and $\left(  \beta,T^{\prime}\right)  $ that the sums agree
up to multiplication by a constant. In this section we present the structure
of Jack polynomials, the assignment of unique labels, and orthogonality
properties (henceforth, unmodified \textquotedblleft Jack\textquotedblright%
\ implies symmetry). Multiplication by $L\left(  x\right)  $ will yield
symmetric eigenfunctions of $\mathcal{H}$, the vector-valued wavefunctions.
The results are mostly from \cite[Sec. 5.2]{DL2011}. The Jack polynomials
correspond to certain connected components of the Yang-Baxter graph after the
affine jumps are removed.

\begin{definition}
For $\alpha\in\mathbb{N}_{0}^{N},T\in\mathcal{Y}\left(  \tau\right)  $ define
$\left\lfloor \alpha,T\right\rfloor $ to be the filling of the Ferrers diagram
of $\tau$ obtained by replacing $i$ by $\alpha_{i}^{+}$ in $T$, for all $i$.
\end{definition}

Obviously $\left\lfloor \alpha,T\right\rfloor =\left\lfloor \alpha
^{+},T\right\rfloor $.

\begin{example}
Let $\tau=\left(  3,2\right)  ,\alpha=\left(  1,4,2,0,3\right)  $ and%
\[
T=%
\begin{array}
[c]{ccc}%
5 & 4 & 1\\
3 & 2 &
\end{array}
,\left\lfloor \alpha,T\right\rfloor =%
\begin{array}
[c]{ccc}%
0 & 1 & 4\\
2 & 3 &
\end{array}
.
\]

\end{example}

\begin{proposition}
(\cite[Prop. 5.2]{DL2011}) $\left(  \alpha,T\right)  $ and $\left(
\beta,T^{\prime}\right)  $ are connected by edges and jumps (without regard to
the orientation) if and only if $\left\lfloor \alpha,T\right\rfloor
=\left\lfloor \beta,T^{\prime}\right\rfloor $.
\end{proposition}

From the properties of steps and jumps it follows that the spectral vectors of
$\left(  \alpha,T\right)  $ and $\left(  \beta,T^{\prime}\right)  $ are
permutations of each other. Set $\mathcal{T}\left(  \alpha,T\right)  =\left\{
\left(  \beta,T^{\prime}\right)  :\left\lfloor \beta,T^{\prime}\right\rfloor
=\left\lfloor \alpha,T\right\rfloor \right\}  $, the set of nodes in the
connected component.

A tableau is \textit{column-strict} if the entries are increasing in each
column, and nondecreasing in each row.

\begin{theorem}
For $\left(  \alpha,T\right)  \in\mathbb{N}_{0}^{N}\times\mathcal{Y}\left(
\tau\right)  $ the $\mathrm{span}\left\{  \zeta_{\beta,T^{\prime}}:\left(
\beta,T^{\prime}\right)  \in\mathcal{T}\left(  \alpha,T\right)  \right\}  $
contains a unique nonzero symmetric polynomial if and only if $\left\lfloor
\alpha,T\right\rfloor $ is column-strict.
\end{theorem}

As usual in this context, unique means up to multiplication by a scalar.
Suppose $\lambda\in\mathbb{N}_{0}^{N,+}$ and consider the sum $p=\sum
\limits_{\left(  \beta,T^{\prime}\right)  \in\mathcal{T}\left(  \lambda
,T\right)  }a\left(  \beta,T^{\prime}\right)  \zeta_{\beta,T^{\prime}}$
subject to the conditions $s_{i}p=p$ for $1\leq i<N$ (sufficing for symmetry).

Suppose there is a step or jump $s_{i}$ from $\left(  \beta,T^{\prime}\right)
$ to $\left(  \gamma,T^{\prime\prime}\right)  $ then $\zeta_{\left(
\gamma,T^{\prime\prime}\right)  }=s_{i}\zeta_{\left(  \beta,T^{\prime}\right)
}-b\zeta_{\left(  \beta,T^{\prime}\right)  }$ for some $b$; this implies
$s_{i}\zeta_{\left(  \gamma,T^{\prime\prime}\right)  }=-b\zeta_{\left(
\beta,T^{\prime}\right)  }+\left(  1-b^{2}\right)  \zeta_{\left(
\gamma,T^{\prime\prime}\right)  }$. The condition $s_{i}\left(  a\left(
\beta,T^{\prime}\right)  \zeta_{\beta,T^{\prime}}+a\left(  \gamma
,T^{\prime\prime}\right)  \zeta_{\gamma,T^{\prime\prime}}\right)  =a\left(
\beta,T^{\prime}\right)  \zeta_{\beta,T^{\prime}}+a\left(  \gamma
,T^{\prime\prime}\right)  \zeta_{\gamma,T^{\prime\prime}}$ implies $a\left(
\beta,T^{\prime}\right)  =\left(  1+b\right)  a\left(  \gamma,T^{\prime\prime
}\right)  $. The column strictness hypothesis implies that $b=-1$ can not
occur. The relation is used in an inductive evaluation of $a\left(
\beta,T^{\prime}\right)  $ once the beginning and end have been identified.

\begin{definition}
For $\alpha\in\mathbb{N}^{N}$ and $T\in Y\left(  \tau\right)  $ let
\begin{align*}
\mathrm{inv}\left(  \alpha\right)   &  :=\#\left\{  \left(  i,j\right)  :1\leq
i<j\leq N,\alpha_{i}<\alpha_{j}\right\}  ,\\
\mathrm{inv}\left(  T\right)   &  :=\#\left\{  \left(  i,j\right)  :1\leq
i<j\leq N,c\left(  i,T\right)  \geq c\left(  j,T\right)  +2\right\}  .
\end{align*}

\end{definition}

Thus a step reduces $\mathrm{inv}\left(  \alpha\right)  $ by $1$ (that is,
$\mathrm{inv}\left(  s_{i}\alpha\right)  =\mathrm{inv}\left(  \alpha\right)
-1$) and a jump reduces $\mathrm{inv}\left(  T\right)  $ by $1$ (in Example
\ref{ex(2,1)} $\mathrm{inv}\left(  T_{0}\right)  =1$ and $\mathrm{inv}\left(
T_{1}\right)  =0$). Hence the root $\left(  \alpha,T\right)  $ of
$\in\mathcal{T}\left(  \alpha,T\right)  $ has maximum $\mathrm{inv}\left(
\alpha\right)  +\mathrm{inv}\left(  T\right)  $ and the sink has minimum
$\mathrm{inv}\left(  \alpha\right)  +\mathrm{inv}\left(  T\right)  $. Clearly
$\mathrm{inv}\left(  \alpha\right)  $ is minimized at $\alpha^{+}$ and
maximized at $\alpha^{-}$, the nondecreasing rearrangement of $\alpha$. In
\cite[Def. 5.6]{DL2011} it is shown that there are unique tableaux
$T_{R},T_{S}$ in $\in\mathcal{T}\left(  \alpha,T\right)  $ such that $\left(
\alpha^{-},T_{R}\right)  $ is the root and $\left(  \alpha^{+},T_{S}\right)  $
is the sink (maximizes, respectively minimizes $\mathrm{inv}\left(
\beta\right)  +\mathrm{inv}\left(  T^{\prime}\right)  $ for $\left(
\beta,T^{\prime}\right)  \in\mathcal{T}\left(  \alpha,T\right)  $). The
formulae for $T_{R}$ and $T_{S}$ are (with $\mathcal{T}=\left\lfloor
\alpha,T\right\rfloor $)%
\begin{align}
T_{R}\left(  i,j\right)   &  =\#\left\{  \left(  k,l\right)  :\mathcal{T}%
\left(  k,l\right)  >\mathcal{T}\left(  i,j\right)  \right\} \label{Troot}\\
&  +\#\left\{  \left(  k,l\right)  :\mathcal{T}\left(  k,l\right)
=\mathcal{T}\left(  i,j\right)  ,\left(  l>j\right)  \vee\left(  l=j\wedge
k\geq i\right)  \right\}  ;\nonumber
\end{align}%
\begin{align}
T_{S}\left(  i,j\right)   &  =\#\left\{  \left(  k,l\right)  :\mathcal{T}%
\left(  k,l\right)  >\mathcal{T}\left(  i,j\right)  \right\} \label{Tsink}\\
&  +\#\left\{  \left(  k,l\right)  :\mathcal{T}\left(  k,l\right)
=\mathcal{T}\left(  i,j\right)  ,\left(  k>i\right)  \vee\left(  k=i\wedge
l\geq j\right)  \right\}  .\nonumber
\end{align}

\begin{example}
$\alpha=\left(  3^{3},2^{3},1,0\right)  ,$%
\[
\left\lfloor \alpha,T\right\rfloor =%
\begin{array}
[c]{cccc}%
0 & 2 & 2 & 3\\
1 & 3 & 3 & \\
2 &  &  &
\end{array}
,T_{R}=%
\begin{array}
[c]{cccc}%
8 & 5 & 4 & 1\\
7 & 3 & 2 & \\
6 &  &  &
\end{array}
,T_{S}=%
\begin{array}
[c]{cccc}%
8 & 6 & 5 & 3\\
7 & 2 & 1 & \\
4 &  &  &
\end{array}
.
\]

\end{example}

As motivation for the formulae for $a\left(  \beta,T^{\prime}\right)  $ in the
sum suppose $\beta_{i}<\beta_{i+1}$ (and $\varepsilon=\pm1$) then%
\begin{align*}
\frac{\mathcal{E}_{\varepsilon}\left(  \beta,T\right)  }{\mathcal{E}%
_{\varepsilon}\left(  s_{i}\beta,T\right)  }  &  =1+\frac{\varepsilon\kappa
}{\xi_{i+1}\left(  \beta,T\right)  -\xi_{i}\left(  \beta,T\right)  },\\
\mathcal{E}_{-1}\left(  \beta,T\right)   &  =\left(  1+b\right)
\mathcal{E}_{-1}\left(  s_{i}\beta,T\right)  ,
\end{align*}
where $\zeta_{s_{i}\beta,T}=s_{i}\zeta_{\beta,T}-b\zeta_{\beta,T}$. We
introduce two functions on $\mathcal{Y}\left(  \tau\right)  $ to deal
analogously with jumps:

\begin{definition}
For $T\in\mathcal{Y}\left(  \tau\right)  $ and $\varepsilon=\pm1$ set%
\[
\mathcal{C}_{\varepsilon}\left(  T\right)  =\prod\left\{  1+\frac{\varepsilon
}{c\left(  i,T\right)  -c\left(  j,T\right)  }:1\leq i<j\leq N,c\left(
i,T\right)  \leq c\left(  j,T\right)  -2\right\}  .
\]

\end{definition}

From (\ref{Tnorm}) $\left\langle T,T\right\rangle _{0}=\mathcal{C}_{1}\left(
T\right)  \mathcal{C}_{-1}\left(  T\right)  $. In the jump with $\beta
_{i}=\beta_{i+1},r_{\beta}\left(  i\right)  =j$ and $c\left(  j,T\right)
-c\left(  j+1,T\right)  \geq2$ (as in \ref{jumpT}) $T^{\left(  j\right)  }$
has $j$ and $j+1$ interchanged so that $c\left(  j,T^{\left(  j\right)
}\right)  -c\left(  j+1,T^{\left(  j\right)  }\right)  =c\left(  j+1,T\right)
-c\left(  j,T\right)  \leq-2.$ Then%
\[
\mathcal{C}_{\varepsilon}\left(  T^{\left(  j\right)  }\right)  =\mathcal{C}%
_{\varepsilon}\left(  T\right)  \left(  1+\frac{\varepsilon}{c\left(
j,T^{\left(  j\right)  }\right)  -c\left(  j+1,T^{\left(  j\right)  }\right)
}\right)
\]
so that $\mathcal{C}_{-1}\left(  T^{\left(  j\right)  }\right)  =\mathcal{C}%
_{-1}\left(  T\right)  \left(  1+b\right)  $ where $b=\left(  c\left(
j,T\right)  -c\left(  j+1,T\right)  \right)  ^{-1}$. From these relations it
can be shown:

\begin{proposition}
Suppose $\left(  \alpha,T\right)  \in\mathbb{N}_{0}^{N}\times\mathcal{Y}%
\left(  \tau\right)  $ and $\left\lfloor \alpha,T\right\rfloor $ is
column-strict then%
\[
p=\sum_{\left(  \beta,T^{\prime}\right)  \in\mathcal{T}\left(  \alpha
,T\right)  }\frac{\mathcal{E}_{-1}\left(  \beta,T^{\prime}\right)
}{\mathcal{C}_{-1}\left(  T^{\prime}\right)  }\zeta_{\beta,T^{\prime}}%
\]
is symmetric and nonzero.
\end{proposition}

To proceed with the analysis we impose a normalization and then find a closed
formula for the squared-norm $\left\Vert \cdot\right\Vert ^{2}$. Replace
$\alpha$ by $\lambda=\alpha^{+}$ and use the sink $\left(  \lambda
,T_{S}\right)  $ as normalization by requiring that the coefficient of
$x^{\lambda}v_{T_{S}}$ is $1$. From the sink property and the
$\vartriangleright$-triangularity of the NSJP's it follows that $x^{\lambda
}v_{T_{S}}$ appears only in $\zeta_{\lambda,T_{S}}$ in the sum $p$, with
coefficient $1$. Thus define (for $\lambda\in\mathbb{N}_{0}^{N,+}$)%
\begin{equation}
J_{\lambda,T_{S}}:=\sum_{\left(  \beta,T^{\prime}\right)  \in\mathcal{T}%
\left(  \lambda,T_{S}\right)  }\frac{\mathcal{C}_{-1}\left(  T_{S}\right)
}{\mathcal{C}_{-1}\left(  T^{\prime}\right)  }\mathcal{E}_{-1}\left(
\beta,T^{\prime}\right)  \zeta_{\beta,T^{\prime}}. \label{Jsum}%
\end{equation}
By orthogonality $\left\Vert J_{\lambda,T_{S}}\right\Vert ^{2}=\sum
\limits_{\left(  \beta,T^{\prime}\right)  \in\mathcal{T}\left(  \lambda
,T_{S}\right)  }\left(  \frac{\mathcal{C}_{-1}\left(  T_{S}\right)
}{\mathcal{C}_{-1}\left(  T^{\prime}\right)  }\mathcal{E}_{-1}\left(
\beta,T^{\prime}\right)  \right)  ^{2}\left\Vert \zeta_{\beta,T^{\prime}%
}\right\Vert ^{2}$; fortunately there is a formula without summation. Suppose
$\left(  \beta,T^{\prime}\right)  \in\mathcal{T}\left(  \lambda,T_{S}\right)
$ and $J_{\lambda,T_{S}}=c\sum\limits_{w\in\mathcal{S}_{N}}w\zeta
_{\beta,T^{\prime}}$ then%
\[
\left\Vert J_{\lambda,T_{S}}\right\Vert ^{2}=c\sum\limits_{w\in\mathcal{S}%
_{N}}\left\langle J_{\lambda,T_{S}},w\zeta_{\beta,T^{\prime}}\right\rangle
_{\mathbb{T}}=N!c\left\langle J_{\lambda,T_{S}},\zeta_{\beta,T^{\prime}%
}\right\rangle _{\mathbb{T}};
\]
if we write $J_{\lambda,T_{S}}=\sum\limits_{\left(  \beta,T^{\prime}\right)
\in\mathcal{T}\left(  \lambda,T_{S}\right)  }a\left(  \beta,T^{\prime}\right)
\zeta_{\beta,T^{\prime}}$ then $\left\Vert J_{\lambda,T_{S}}\right\Vert
^{2}=\left(  N!c\right)  a\left(  \beta,T^{\prime}\right)  \left\Vert
\zeta_{\beta,T^{\prime}}\right\Vert ^{2}$ and $c$ can be determined by careful
choice of $\left(  \beta,T^{\prime}\right)  $. Consider the stabilizer group
$G_{\lambda,T_{S}}$ of $\zeta_{\lambda,T_{S}}$ ($w\in G_{\lambda,T_{S}}$
implies $w\zeta_{\lambda,T_{S}}=\zeta_{\lambda,T_{S}}$). The group is
generated by $\left\{  s_{i}:\lambda_{i}=\lambda_{i+1},\mathrm{rw}\left(
i,T_{S}\right)  =\mathrm{rw}\left(  i+1,T_{S}\right)  \right\}  $. It was
shown (\cite[Prop. 5.11]{DL2011}) that the coefficient of $\zeta
_{\lambda,T_{S}}$ in $\sum\limits_{w\in\mathcal{S}_{N}}w\zeta_{\lambda
^{-},T_{R}}$ is $\#G_{\lambda,T_{S}}$, hence that $\sum\limits_{w\in
\mathcal{S}_{N}}w\zeta_{\lambda^{-},T_{R}}=\#G_{\lambda,T_{S}}J_{\lambda
,T_{S}}$. We deduce%
\[
\left\Vert J_{\lambda,T_{S}}\right\Vert ^{2}=\left(  \frac{N!}{\#G_{\lambda
,T_{S}}}\right)  \frac{\mathcal{C}_{-1}\left(  T_{S}\right)  }{\mathcal{C}%
_{-1}\left(  T_{R}\right)  }\mathcal{E}_{-1}\left(  \lambda^{-},T_{R}\right)
\left\Vert \zeta_{\lambda^{-},T_{R}}\right\Vert ^{2}.
\]
Also%
\[
\left\Vert \zeta_{\lambda^{-},T_{R}}\right\Vert ^{2}=\left(  \mathcal{E}%
_{-1}\left(  \lambda^{-},T_{R}\right)  \mathcal{E}_{1}\left(  \lambda
^{-},T_{R}\right)  \right)  ^{-1}\left\Vert \zeta_{\lambda,T_{R}}\right\Vert
^{2}%
\]
and $\dfrac{\left\Vert \zeta_{\lambda,T_{R}}\right\Vert ^{2}}{\left\langle
T_{R},T_{R}\right\rangle _{0}}=\dfrac{\left\Vert \zeta_{\lambda,T_{S}%
}\right\Vert ^{2}}{\left\langle T_{S},T_{S}\right\rangle _{0}}$ (from Theorem
\ref{Pnorm1}). From (\ref{Tnorm}) it follows that $\dfrac{\left\langle
T_{R},T_{R}\right\rangle _{0}}{\left\langle T_{S},T_{S}\right\rangle _{0}%
}=\dfrac{\mathcal{C}_{1}\left(  T_{R}\right)  \mathcal{C}_{-1}\left(
T_{R}\right)  }{\mathcal{C}_{1}\left(  T_{S}\right)  \mathcal{C}_{-1}\left(
T_{S}\right)  }$, and $\dfrac{N!}{\#G_{\lambda,T_{S}}}=\#\mathcal{T}%
_{\lambda,T_{S}}$. To summarize:

\begin{theorem}
\label{JTdef}Suppose $\left(  \lambda,T\right)  \in\mathbb{N}_{0}^{N,+}%
\times\mathcal{Y}\left(  \tau\right)  $ and $\left\lfloor \lambda
,T\right\rfloor $ is column-strict. Define $T_{R}$ and $T_{S}$ by formulae
(\ref{Troot}) and (\ref{Tsink}) then%
\[
\left\Vert J_{\lambda,T_{S}}\right\Vert ^{2}=\left(  \#\mathcal{T}%
_{\lambda,T_{S}}\right)  \frac{\mathcal{C}_{1}\left(  T_{R}\right)
}{\mathcal{C}_{1}\left(  T_{S}\right)  \mathcal{E}_{1}\left(  \lambda
^{-},T_{R}\right)  }\left\Vert \zeta_{\lambda,T_{S}}\right\Vert ^{2}.
\]

\end{theorem}

Suppose $\left\lfloor \lambda,T_{S}\right\rfloor ,\left\lfloor \lambda
^{\prime},T_{S}\right\rfloor $ are unequal column-strict tableaux. By
definition $\mathcal{T}\left(  \lambda,T_{S}\right)  \cap\mathcal{T}\left(
\lambda^{\prime},T_{S}\right)  =\varnothing$ and from the mutual orthogonality
of the terms in the sums (\ref{Jsum}) it follows that $\left\langle
J_{\lambda,T_{S}},J_{\lambda^{\prime},T_{S}}\right\rangle _{\mathbb{T}}=0$.

Multiplication of $J_{\lambda,T_{S}}$ by $e_{N}^{m}$ produces the Jack
polynomial $J_{\lambda+m\boldsymbol{1},T_{S}}$ (see Definition \ref{eNmult});
here $m=-1,-2,\ldots$is valid and defines Jack Laurent polynomials. The
expression $e_{N}^{m}J_{\lambda,T_{S}}$ is made unique by the requirement
$\lambda_{N}=0$.

We see that there is a unique symmetric polynomial $p_{\lambda,T}$ of minimum
degree: the tableau $\left\lfloor \lambda,T\right\rfloor $ has the entry $i-1$
in each box in row $\#i$. The degree is $n\left(  \tau\right)  =\sum_{i\geq
1}\left(  i-1\right)  \tau_{i}$.

\begin{remark}
A weak reverse tableau of shape $\tau$ and weight $n$ has its entries
nondecreasing in each row and in each column and the sum of the entries is
$n$. Such a tableau can be transformed to a column-strict tableau by adding
$i-1$ to each entry in row $\#i$ for $1\leq i\leq\ell\left(  \tau\right)  .$
The number of the weak reverse tableaux is the coefficient of $z^{n}$ in
$H_{\tau}\left(  z\right)  :=\prod_{\left(  i,j\right)  \in\tau}\left(
1-z^{h\left(  i,j\right)  }\right)  ^{-1}$ (see \cite[p.379]{S1999}, $h\left(
i,j\right)  $ from Formula (\ref{hllt})). Thus the number of Jack polynomials
of degree $n$ is the coefficient of $z^{n}$ in $z^{n\left(  \tau\right)
}H_{\tau}\left(  z\right)  $. For example take $\tau=\left(  3,2\right)  $
then $z^{n\left(  \tau\right)  }H_{\tau}\left(  z\right)  =z^{2}\left\{
\left(  1-z\right)  ^{2}\left(  1-z^{2}\right)  \left(  1-z^{3}\right)
\left(  1-z^{4}\right)  \right\}  ^{-1}$. The analogous number when
$\lambda_{N}=0$ is the coefficient of $z^{n}$ in $\left(  1-z^{N}\right)
z^{n\left(  \tau\right)  }H_{\tau}\left(  z\right)  $.
\end{remark}

If $\left(  \beta,T^{\prime}\right)  \in\mathcal{T}\left(  \lambda
,T_{S}\right)  $ then the spectral vector $\xi_{\beta,T^{\prime}}$ is a
permutation of $\xi_{\lambda,T_{S}}$ thus $\sum_{i=1}^{N}\mathcal{U}_{i}%
^{m}J_{\lambda,T_{S}}=\sum_{j=1}^{N}\xi_{j}\left(  \lambda,T_{S}\right)
^{m}J_{\lambda,T_{S}}$ for $m=1,2,3,\ldots$. and%
\begin{equation}
\sum_{i=1}^{N}\left(  \mathcal{U}_{i}-1-\kappa\gamma\right)  ^{2}%
J_{\lambda,T_{S}}=\sum_{i=1}^{N}\left(  \lambda_{i}+\kappa\left(  c\left(
i,T_{S}\right)  -\gamma\right)  \right)  ^{2}J_{\lambda,T_{S}}.
\label{eigvalH}%
\end{equation}
Set $S_{2}:=\sum_{i=1}^{\ell\left(  \tau\right)  }c\left(  i,T_{S}\right)
^{2}=\frac{1}{6}\sum_{i=1}^{\ell\left(  \tau\right)  }\tau_{i}\left\{  \left(
\tau_{i}-1\right)  \left(  \tau_{i}-2\right)  -6\left(  \tau_{i}-i\right)
\left(  i-1\right)  \right\}  $. The eigenvalue can be written as $\sum
_{i=1}^{N}\lambda_{i}^{2}+2\kappa\sum_{i=1}^{N}\lambda_{i}\left(  c\left(
i,T_{S}\right)  -\gamma\right)  +\kappa^{2}\left(  S_{2}-N\gamma^{2}\right)
.$ In the trivial case $\tau=\left(  N\right)  $ the last term becomes
$\frac{1}{12}\kappa^{2}N\left(  N^{2}-1\right)  $. The effect of multiplying
by $e_{N}^{m}$ on the eigenvalue is
\begin{gather*}
\sum_{i=1}^{N}\left(  \mathcal{U}_{i}-1-\kappa\gamma\right)  ^{2}e_{N}%
^{m}J_{\lambda,T_{S}}=\sum_{i=1}^{N}\left(  \lambda_{i}+m+\kappa\left(
c\left(  i,T_{S}\right)  -\gamma\right)  \right)  ^{2}e_{N}^{m}J_{\lambda
,T_{S}}\\
=\left\{  \sum_{i=1}^{N}\left(  \lambda_{i}+\kappa\left(  c\left(
i,T_{S}\right)  -\gamma\right)  \right)  ^{2}+2m\sum_{i=1}^{N}\lambda
_{i}+Nm^{2}\right\}  e_{N}^{m}J_{\lambda,T_{S}}.
\end{gather*}
This is minimized over $m$ when $m$ is the nearest integer to $-\sum_{i=1}%
^{N}\lambda_{i}/N$.

\section{\label{symmwav}Symmetric wavefunctions}

In the notation of Section \ref{symmvp} there is a set of mutually orthogonal
wavefunctions $L\left(  x\right)  J_{\lambda,T_{S}}\left(  x\right)  $ such
that
\[
\mathcal{H}L\left(  x\right)  J_{\lambda,T_{S}}\left(  x\right)  =\sum
_{i=1}^{N}\left(  \lambda_{i}+\kappa\left(  c\left(  i,T_{S}\right)
-\gamma\right)  \right)  ^{2}L\left(  x\right)  J_{\lambda,T_{S}}\left(
x\right)
\]
by Theorem \ref{LHL}. Thus%
\[
\frac{1}{\left\Vert J_{\lambda,T_{S}}\right\Vert ^{2}}\left(  L\left(
x\right)  J_{\lambda,T_{S}}\left(  x\right)  \right)  ^{\ast}L\left(
x\right)  J_{\lambda,T_{S}}\left(  x\right)
\]
is a probability density function on $\mathbb{T}^{N}$. Since multiplication of
$L\left(  x\right)  J_{\lambda,T_{S}}\left(  x\right)  $ by powers of $e_{N}$
does not change the density function, we can assume $\lambda_{N}=0$. In
contrast to the scalar case the matrix $L\left(  x\right)  $ has singularities
of order $\left\vert x_{i}-x_{j}\right\vert ^{\pm\kappa}$ in neighborhoods of
points $x$ with $x_{i}=x_{j}$ and all other $x_{k}$ being pairwise distinct.
Nevertheless we can show that the symmetric wavefunctions are bounded in such
sets when $0<\kappa<\frac{1}{h_{\tau}}$. By the invariance it suffices to
prove this near $\left\{  x:x_{N-1}=x_{N}\right\}  $ in the fundamental
chamber. More precisely let $\delta>0$ and define
\[
\Omega_{\delta}:=\left\{  x\in\mathcal{C}_{0}:1\leq i<j\leq N-1\Longrightarrow
\left\vert x_{i}-x_{j}\right\vert \geq\delta,\left\vert x_{N}-x_{1}\right\vert
\geq\delta\right\}  .
\]

\begin{theorem}
Suppose $0<\kappa<\frac{1}{h_{\tau}}$ and $\left(  \lambda,T_{S}\right)  $ is
as in Theorem \ref{JTdef} then $L\left(  x\right)  J_{\lambda,T_{S}}\left(
x\right)  $ is uniformly bounded in $\Omega_{\delta}$.
\end{theorem}

The proof depends on a power series formulation for $L$ proven in \cite[Sec.
5]{D2017}. Arrange $\mathcal{Y}\left(  \tau\right)  $ linearly listing the
tableaux $T$ with $c\left(  N-1,T\right)  =-1$ first (that is, $\mathrm{cm}%
\left(  N-1,T\right)  =1,\mathrm{rw}\left(  N-1,T\right)  =2$). This results
in the matrix representation of $\tau\left(  N-1,N\right)  $ being
\[
\left[
\begin{array}
[c]{cc}%
-I_{m_{\tau}} & O\\
O & I_{n_{\tau}-m_{\tau}}%
\end{array}
\right]  ,
\]
where $n_{\tau}:=\dim V_{\tau}=\#\mathcal{Y}\left(  \tau\right)  $ and
$m_{\tau}$ is given by $\mathrm{tr}\left(  \tau\left(  N-1,N\right)  \right)
=n_{\tau}-2m_{\tau}$. From the sum $\sum_{i<j}\tau\left(  i,j\right)
=S_{1}\left(  \tau\right)  I$ it follows that $\binom{N}{2}\mathrm{tr}\left(
\tau\left(  N-1,N\right)  \right)  =S_{1}\left(  \tau\right)  n_{\tau}$ and
$m_{\tau}=n_{\tau}\left(  \frac{1}{2}-\frac{S_{1}\left(  \tau\right)
}{N\left(  N-1\right)  }\right)  $. (The traces of the transpositions are
equal because they are conjugate to each other.) The property of a matrix
commuting or anti-commuting with $\sigma:=\tau\left(  N-1,N\right)  $ is used
in the argument; to express this neatly we introduce the $\sigma$-block
decomposition $\left(  m_{\tau}+\left(  n_{\tau}-m_{\tau}\right)  \right)
\times\left(  m_{\tau}+\left(  n_{\tau}-m_{\tau}\right)  \right)  $ of a
matrix%
\[
\alpha=\left[
\begin{array}
[c]{cc}%
\alpha_{11} & \alpha_{12}\\
\alpha_{21} & \alpha_{22}%
\end{array}
\right]  .
\]
Then $\sigma\alpha\sigma=\alpha$ if and only if $\alpha_{12}=O=\alpha_{21}$
and $\sigma\alpha\sigma=-\alpha$ if and only if $\alpha_{11}=O=\alpha_{22}$.
For $z_{1},z_{2}\in\mathbb{C}$ let
\[
\rho\left(  z_{1},z_{2}\right)  :=\left[
\begin{array}
[c]{cc}%
z_{1}I_{m_{\tau}} & O\\
O & z_{2}I_{n_{\tau}-m_{\tau}}%
\end{array}
\right]  .
\]
We showed in \cite[Sec. 5]{D2017} that there exist matrix coefficients
$\alpha_{n}\left(  x^{\prime}\right)  $ with $x^{\prime}:=\left(  x_{1}%
,\ldots,x_{N-2},\frac{x_{N-1}+x_{N}}{2}\right)  $, analytic on the closure of
$\Omega_{\delta}\cup\Omega_{\delta}\left(  N-1,N\right)  $, such that (with
$z:=\frac{x_{N}-x_{N-1}}{2}$)%
\[
L\left(  x\right)  =\left(  x_{1}x_{2}\cdots x_{N}\right)  ^{-\gamma\kappa
}\rho\left(  z^{-\kappa},z^{\kappa}\right)  \sum_{n=0}^{\infty}\alpha
_{n}\left(  x^{\prime}\right)  z^{n}%
\]
and $\sigma\alpha_{n}\left(  x^{\prime}\right)  \sigma=\left(  -1\right)
^{n}\alpha_{n}\left(  x^{\prime}\right)  .$ In particular the $\sigma$-block
decomposition of $\alpha_{0}\left(  x^{\prime}\right)  $ is $\left[
\begin{array}
[c]{cc}%
\alpha_{0,11}\left(  x^{\prime}\right)  & O\\
O & \alpha_{0,22}\left(  x^{\prime}\right)
\end{array}
\right]  $. The series converges absolutely for $\frac{1}{2}\left\vert
x_{N}-x_{N-1}\right\vert <\min\limits_{1\leq j\leq N-2}\left\vert x_{j}%
-\frac{x_{N-1}+x_{N}}{2}\right\vert $.

Suppose $J_{\lambda,T_{S}}$ is nonzero as in (\ref{Jsum}) then express
$J_{\lambda,T_{S}}\left(  x\right)  \allowbreak=\sum\limits_{T\in
\mathcal{Y}\left(  \tau\right)  }\frac{1}{\left\langle T,T\right\rangle
_{0}^{1/2}}p_{T}\left(  x\right)  \otimes T$ where each $p_{T}\left(
x\right)  $ is a scalar polynomial. If $c\left(  N-1,T\right)  =-1$ then
$\tau\left(  N-1,N\right)  T=-T$ and the relation $\left(  N-1,N\right)
J_{\lambda,T_{S}}=J_{\lambda,T_{S}}$ implies $p_{T}\left(  x\left(
N-1,N\right)  \right)  =-p\left(  x\right)  $ (see Proposition. \ref{sympol}).
In the order of the basis chosen above the first $m_{\tau}$ coefficients of
$J_{\lambda,T_{S}}$ are divisible by $x_{N}-x_{N-1}$. Write $J_{\lambda,T_{S}%
}$ as a column vector $\left[  p_{1}^{tr},p_{2}^{tr}\right]  ^{tr}$ with
$p_{1}$ consisting of the first $m_{\tau}$ coordinates. Suppose $\kappa>0$
then the dominant part of $\left(  L\left(  x\right)  J_{\lambda,T_{S}}\left(
x\right)  \right)  ^{\ast}L\left(  x\right)  J_{\lambda,T_{S}}\left(
x\right)  $ is%
\begin{align*}
&  \left[  p_{1}^{\ast},p_{2}^{\ast}\right]  \alpha_{0}\left(  x^{\prime
}\right)  ^{\ast}\rho\left(  \left\vert z\right\vert ^{-2\kappa},\left\vert
z\right\vert ^{2\kappa}\right)  \alpha_{0}\left(  x^{\prime}\right)  \left[
p_{1}^{tr},p_{2}^{tr}\right]  ^{tr}\\
&  =\left\vert z\right\vert ^{-2\kappa}p_{1}\left(  x\right)  ^{\ast}%
\alpha_{0,11}\left(  x^{\prime}\right)  ^{\ast}\alpha_{0,11}\left(  x^{\prime
}\right)  p_{1}\left(  x\right)  +\left\vert z\right\vert ^{2\kappa}%
p_{2}\left(  x\right)  ^{\ast}\alpha_{0,22}\left(  x^{\prime}\right)  ^{\ast
}\alpha_{0,22}\left(  x^{\prime}\right)  p_{2}\left(  x\right)  ;
\end{align*}
(with $z=\frac{x_{N}-x_{N-1}}{2}$) the omitted terms are of order $\left\vert
x_{N}-x_{N-1}\right\vert ^{1-2\left\vert \kappa\right\vert }$. Each component
of $p_{1}\left(  x\right)  $ is divisible by $x_{N}-x_{N-1}$ and thus the
first part of the expression is of order $\left\vert x_{N}-x_{N-1}\right\vert
^{2-2\kappa}$ and the second part is of order $\left\vert x_{N}-x_{N-1}%
\right\vert ^{2\kappa}$. Hence $\left(  L\left(  x\right)  J_{\lambda,T_{S}%
}\left(  x\right)  \right)  ^{\ast}L\left(  x\right)  J_{\lambda,T_{S}}\left(
x\right)  $ is bounded on $\Omega_{\delta}$. It may be suspected that the
bound holds for all of $\mathbb{T}^{N}$ but the behaviour of $L\left(
x\right)  $ near a point with multiple repeated entries (say $x_{N-2}%
=x_{N-1}=x_{N}$) is complicated and the power series method used here does not apply.

\subsection{Minimal degree symmetric polynomials}

To produce the minimal degree $J_{\lambda,T_{S}}$, or equivalently, the
minimal degree column-strict tableau of shape $\tau,$ the entries in row $\#i$
all equal $i-1$. The corresponding $T_{S}$ has the numbers $N,N-1,\ldots,2,1$
entered row-by-row, and $T_{R}=T_{S}$. With $l=\ell\left(  \tau\right)  $ and
using superscripts to denote multiplicity $\lambda=\left(  \left(  l-1\right)
^{\tau_{l}},\left(  l-2\right)  ^{\tau_{l-1}},\ldots,1^{\tau_{2}},0^{\tau_{1}%
}\right)  $. There is an explicit formula for $J_{\lambda,T_{S}}$. It is
derived from Proposition \ref{sympol} and involves the Specht polynomials. For
$1\leq n_{1}\leq n_{2}\leq N$ define the alternating polynomial $a\left(
x;n_{1},n_{2}\right)  =\prod\limits_{n_{1}\leq i<j\leq n_{2}}\left(
x_{i}-x_{j}\right)  $ (the empty product $a\left(  x;n_{1},n_{1}\right)  =1$).
Denote the transpose of the partition $\tau$ by $\tau^{\prime}$ then $\tau
_{1}^{\prime}=\ell\left(  \tau\right)  $. Form $p_{T_{0}}$ as the product of
the alternating polynomials for each column of $T_{0}$: that is, set
$k_{j}=N-\sum_{i=1}^{j}\tau_{j}^{\prime},0\leq j\leq\tau_{1}$ and
\[
p_{T_{0}}\left(  x\right)  :=\prod\limits_{j=1}^{\tau_{1}}a\left(
x;k_{j}+1,k_{j-1}\right)  .
\]
Then the other polynomials $p_{T}$ are produced by the formulae in the
Proposition and the minimal $J_{\lambda,T_{S}}=\left\langle T_{S}%
,T_{S}\right\rangle _{0}\sum\limits_{T^{\prime}\in\mathcal{Y}\left(
\tau\right)  }\frac{1}{\left\langle T^{\prime},T^{\prime}\right\rangle _{0}%
}p_{T^{\prime}}\left(  x\right)  \otimes T^{\prime}$ (see \cite[Sec.
5.4]{DL2011}). For example let $\tau=\left(  2,2\right)  $ then%
\begin{equation}
T_{0}=%
\begin{array}
[c]{cc}%
4 & 2\\
3 & 1
\end{array}
,T_{1}=%
\begin{array}
[c]{cc}%
4 & 3\\
2 & 1
\end{array}
,\left\lfloor \lambda,T_{1}\right\rfloor =%
\begin{array}
[c]{cc}%
0 & 0\\
1 & 1
\end{array}
, \label{tb(,2,2)}%
\end{equation}%
\[
p_{T_{0}}=\left(  x_{3}-x_{4}\right)  \left(  x_{1}-x_{2}\right)  ,p_{T_{1}%
}=s_{2}p_{T_{1}}-\frac{1}{2}p_{T_{0}}=x_{1}x_{2}+x_{3}x_{4}-\frac{1}{2}\left(
x_{1}+x_{2}\right)  \left(  x_{3}+x_{4}\right)  ,
\]
because $b=\left(  c\left(  2,T_{0}\right)  -c\left(  3,T_{0}\right)  \right)
^{-1}=\frac{1}{2}$. Observe that $s_{1}p_{T_{1}}=p_{T_{1}}=s_{3}p_{T_{1}}$.
Also $\left\langle T_{0},T_{0}\right\rangle _{0}=1$ and $\left\langle
T_{1},T_{1}\right\rangle _{0}=\frac{3}{4}$ (see (\ref{Tnorm}), $\lambda
=\left(  1,1,0,0\right)  $ and $J_{\lambda,T_{1}}=\frac{3}{4}p_{T_{0}}\otimes
T_{0}+p_{T_{1}}\otimes T_{1}$.

There is a formula for the minimal $\left(  \lambda,T_{S}\right)  $ proven in
\cite[Thm. 8]{D2010}:
\[
\left\Vert J_{\lambda,T_{S}}\right\Vert ^{2}=N!\left(  \prod\limits_{i=1}%
^{\ell\left(  \tau\right)  }\tau_{i}!\right)  ^{-1}\left\langle T_{S}%
,T_{S}\right\rangle _{0}\prod\limits_{\left(  i,j\right)  \in\tau}%
\frac{\left(  1-\kappa h\left(  i,j\right)  \right)  _{\mathrm{leg}\left(
i,j\right)  }}{\left(  1+\kappa\left(  j-i\right)  \right)  _{i-1}},
\]
where $\mathrm{leg}\left(  i,j\right)  :=\#\left\{  l:l>i,\tau_{l}\geq
j\right\}  $ (equivalently $\tau_{j}^{\prime}-i$), and $\left(  m\right)
_{n}$ is the Pochhammer symbol $\prod_{j=1}^{n}\left(  m+j-1\right)  $.

The energy eigenvalue $\sum_{i=1}^{N}\lambda_{i}^{2}+2\kappa\sum_{i=1}%
^{N}\lambda_{i}\left(  c\left(  i,T_{S}\right)  -\gamma\right)  +\kappa
^{2}\left(  S_{2}-N\gamma^{2}\right)  $ from (\ref{eigvalH}) specializes to%
\[
\sum_{i=1}^{\ell\left(  \tau\right)  }\left(  i-1\right)  ^{2}\tau_{i}%
+\kappa\sum_{i=1}^{\ell\left(  \tau\right)  }\tau_{i}\left(  i-1\right)
\left(  \tau_{i}+1-2i-2\gamma\right)  +\kappa^{2}\left(  S_{2}-N\gamma
^{2}\right)  ,
\]
for the minimal degree $J_{\lambda,T_{S}}$.

\subsection{Example}

Let $N=4,\tau=\left(  2,2\right)  ,$ $\mathcal{Y}\left(  \tau\right)
=\left\{  T_{0},T_{1}\right\}  $ see (\ref{tb(,2,2)}), $h_{\tau}=3$. The
system (\ref{Lsys}) can be reduced to a hypergeometric equation in one
variable by the substitution $\zeta\left(  x\right)  =\dfrac{\left(
x_{1}-x_{2}\right)  \left(  x_{3}-x_{4}\right)  }{\left(  x_{1}-x_{3}\right)
\left(  x_{2}-x_{4}\right)  }$. The bounds $0<\zeta\left(  x\right)  <1$ hold
in the fundamental domain $\mathcal{C}_{0}$. A fundamental solution $L_{F}$ is
given in terms of the functions%
\begin{align*}
g_{1}\left(  \kappa;\zeta\right)   &  :=~_{2}F_{1}\left(
\genfrac{}{}{0pt}{}{-\kappa,\kappa}{2\kappa}%
;\zeta\right)  ,\\
g_{2}\left(  \kappa;\zeta\right)   &  :=\frac{\kappa\zeta}{1+2\kappa}%
~_{2}F_{1}\left(
\genfrac{}{}{0pt}{}{1+\kappa,1-\kappa}{2+2\kappa}%
;\zeta\right)  ;
\end{align*}

\[
L_{F}\left(  \zeta\right)  :=%
\begin{bmatrix}
\zeta^{-\kappa}\left(  1-\zeta\right)  ^{\kappa} & 0\\
0 & \zeta^{\kappa}\left(  1-\zeta\right)  ^{-\kappa}%
\end{bmatrix}%
\begin{bmatrix}
g_{1}\left(  -\kappa;\zeta\right)  & \frac{\sqrt{3}}{2}g_{2}\left(
-\kappa;\zeta\right) \\
-\frac{\sqrt{3}}{2}g_{2}\left(  \kappa;\zeta\right)  & g_{1}\left(
\kappa;\zeta\right)
\end{bmatrix}
.
\]
Then the solution $L\left(  x\right)  $ which satisfies (\ref{goodL}) is up to
a positive multiplicative constant%
\[
L\left(  x\right)  =\left[
\begin{array}
[c]{cc}%
\gamma\left(  \kappa\right)  ^{1/2} & 0\\
0 & \gamma\left(  -\kappa\right)  ^{1/2}%
\end{array}
\right]  L_{F}\left(  \zeta\left(  x\right)  \right)  ,
\]
where $\gamma\left(  \kappa\right)  :=\frac{\Gamma\left(  1+2\kappa\right)
^{2}}{\Gamma\left(  1+\kappa\right)  \Gamma\left(  1+3\kappa\right)  }$;
observe that $\gamma\left(  \kappa\right)  >0$ for $\kappa>-\frac{1}{3}$. As
yet the problem of determining the normalizing constant is still open. The
purpose of the example is to demonstrate the qualitative difference of
$L\left(  x\right)  $ from the scalar case, so the underlying computations are
not presented here. The method uses the known transformation properties of the
hypergeometric series for the change-of-variable $\zeta\rightarrow1-\zeta$ ,
since $\zeta\left(  xw_{0}\right)  =1-\zeta\left(  x\right)  $ where
$w_{0}=\left(  1,2,3,4\right)  $, a $4$-cycle.

\end{document}